\documentclass[subeqn,11pt, a4paper]{amsart}

\usepackage[margin=35mm]{geometry}

\usepackage{amssymb,amsmath,amsfonts,amsthm, amscd, mathrsfs, helvet, mathtools}
\usepackage{bm, stmaryrd}
\usepackage{graphicx, verbatim}
\usepackage[mathcal]{euscript}
\usepackage{color}
\usepackage[all]{xy}
\usepackage{relsize}

\usepackage{tikz}
\usetikzlibrary{cd}
\usetikzlibrary{matrix,calc}
\usetikzlibrary{decorations.markings,shapes.geometric,shapes.misc}
\usetikzlibrary{decorations.pathreplacing,positioning}
\usetikzlibrary{arrows}

\tikzset{cross/.style={cross out, draw=black, minimum size=2*(#1-\pgflinewidth), inner sep=0pt, outer sep=0pt}, cross/.default={1pt}}

\tikzset{cross/.style={cross out, draw=black, minimum size=2*(#1-\pgflinewidth), inner sep=0pt, outer sep=0pt}, cross/.default={1pt}}

\allowdisplaybreaks[1]

\usepackage{pict2e}

\usepackage{hyperref}
\hypersetup{
     colorlinks=false,         
     linkcolor=darkred,
     citecolor=blue,
}



\definecolor{myPurple}{rgb}{0.5,0.1,0.6}
\definecolor{myOrange}{rgb}{1.0,0.5,0.0}
\definecolor{myRed}{rgb}{1.0,0.0,0.0}
\definecolor{myGreen}{rgb}{0.0,0.5,0.0}
\definecolor{LatexBlue}{rgb}{0.211765,0.227451,0.666667}
\definecolor{myBlue}{rgb}{0.0,0.0,1.0}
\definecolor{myBlack}{rgb}{0.0,0.0,0.0}
\definecolor{myGray}{rgb}{0.3,0.3,0.3}

\theoremstyle{plain}
\newtheorem{theorem}{Theorem}[section]
\newtheorem*{theorem*}{Theorem}
\newtheorem{proposition}[theorem]{Proposition}
\newtheorem*{proposition*}{Proposition}
\newtheorem{lemma}[theorem]{Lemma}
\newtheorem{corollary}[theorem]{Corollary}

\theoremstyle{definition}
\newtheorem{definition}[theorem]{Definition}
\newtheorem{examplex}[theorem]{Example}
\newtheorem{remarkx}[theorem]{Remark}
\newenvironment{example}
  {\pushQED{\qed}\examplex}
  {\popQED\endexamplex}
\newenvironment{remark}
  {\pushQED{\qed}\remarkx}
  {\popQED\endremarkx}

\newcommand{\tensor}[1]{{\mathfrak{#1}}}

\DeclareMathOperator{\Map}{Map}

\newcommand{\longhookrightarrow}{\lhook\joinrel\relbar\joinrel\rightarrow}

\def\g{\mathfrak{g}}

\def\id{\textup{id}}

\def\k{\mathfrak{k}}

\def\L{\mathcal{L}}
\def\T{\mathcal{T}}

\def\CS{\mathrm{CS}}

\newcommand{\lau}[1]{(\kern-.2em( #1 )\kern-.2em)}

\def\CC{\mathbb{C}}
\def\CP{\mathbb{C}P^1}

\def\RR{\mathbb{R}}

\def\ZZ{\mathbb{Z}}

\def\M{\mathcal M}

\def\ii{{\rm i}}

\def\1{\tensor{1}}
\def\2{\tensor{2}}
\def\3{\tensor{3}}
\def\4{\tensor{4}}

\numberwithin{equation}{section}

\begin{document}

\title[Homotopical analysis of 4d CS and integrable field theories]{Homotopical analysis of 4d Chern-Simons theory\\[1mm]
and integrable field theories}

\author{Marco Benini}
\address{Dipartimento di Matematica, Universit\`{a} di Genova, Via Dodecaneso 35, 16146 Genova, Italy, and INFN, Sezione di Genova, Via Dodecaneso 33, 16146 Genova, Italy}
\email{benini@dima.unige.it}
\author{Alexander Schenkel}
\address{School of Mathematical Sciences, University of Nottingham, University Park, Nottingham NG7 2RD, United Kingdom}
\email{alexander.schenkel@nottingham.ac.uk}
\author{Beno\^{\i}t Vicedo}
\address{Department of Mathematics, University of York, York YO10 5DD, United Kingdom}
\email{benoit.vicedo@gmail.com}

\begin{abstract}
This paper provides a detailed study of $4$-dimensional Chern-Simons theory on $\mathbb{R}^2\times \mathbb{C}P^1$ for an arbitrary meromorphic $1$-form $\omega$ on $\mathbb{C}P^1$. Using techniques from homotopy theory, the behaviour under finite gauge transformations of a suitably regularised version of the action proposed by Costello and Yamazaki is investigated. Its gauge invariance is related to boundary conditions on the surface defects located at the poles of $\omega$ that are determined by isotropic Lie subalgebras of a certain defect Lie algebra. The groupoid of fields satisfying such a boundary condition is proved to be equivalent to a groupoid that implements the boundary condition through a homotopy pullback, leading to the appearance of edge modes. The latter perspective is used to clarify how integrable field theories arise from $4$-dimensional Chern-Simons theory.
\end{abstract}

\maketitle

\setcounter{tocdepth}{1}
\tableofcontents

\input{epsf}

\section{Introduction}\label{sec: intro}
Integrable field theories in $2$ dimensions are characterised by the existence of 
an on-shell flat connection that depends meromorphically on an auxiliary Riemann surface, 
typically the Riemann sphere $\CP$. Such a \emph{Lax connection} is often found by some 
clever guesswork, hence its origin is usually rather mysterious.

More recently, new approaches have been developed that provide very interesting
algebraic and/or geometric explanations for the origin of Lax connections.
From an algebraic perspective, $2$-dimensional classical integrable field 
theories can be described in the Hamiltonian formalism as particular 
representations of Gaudin models associated with affine Kac-Moody algebras \cite{Vicedo:2017cge}.
From a geometric perspective, it was realised by Costello and Yamazaki \cite{Costello:2019tri} 
that classical integrable field theories on a $2$-dimensional manifold $\Sigma$ arise as specific 
solutions to a $4$-dimensional generalisation of Chern-Simons theory, see also 
\cite{Costello:2013zra, Costello:2013sla, Witten:2016spx, Costello:2017dso, Costello:2018gyb}
for earlier works on this subject and \cite{Vicedo:2019dej} for a relation to affine Gaudin models.
The Lagrangian of the latter theory is given by $\omega \wedge \CS(A)$, where $\omega$ is a 
(fixed) meromorphic $1$-form on $\CP$ and $\CS(A)$ is the Chern-Simons $3$-form for a 
$\g$-valued $1$-form $A$ living on the product manifold $X = \Sigma \times C$, where $C$ is
the Riemann sphere with the zeroes of $\omega$ removed to allow $A$ to have singularities there.
In this approach, different integrable field theories on $\Sigma$ are obtained from
different choices of meromorphic $1$-forms $\omega$ together with suitable 
boundary conditions on the surface defects $\Sigma \times \{ x \}\subset 
\Sigma \times \CP$  located at the poles $x$ of $\omega$.
In particular, the equations of motion for the $\g$-valued $1$-form $A$ in the bulk, i.e. away from the poles of $\omega$, admit meromorphic solutions with poles at the zeroes of $\omega$, which correspond to the Lax connection of the integrable field theory.

\medskip

The goal of the present paper is twofold. First, we provide a detailed and rigorous
study of the $4$-dimensional Chern-Simons action of \cite{Costello:2019tri}, 
its invariance under \emph{finite} gauge transformations, and the structure of 
boundary conditions on the surface defects. 
For this we consider an arbitrary meromorphic $1$-form $\omega$ on $\CP$, 
with an arbitrary finite set of poles $\bm z \subset \CP$ with each pole 
$x \in \bm z$ having an arbitrary order $n_x \in \ZZ_{\geq 1}$, 
which generalises considerably the cases of simple and double poles 
studied previously, see e.g.\ \cite{Costello:2019tri, Delduc:2019whp}. 
(We would like to emphasise that, in the presence of higher order poles, 
the $4$-dimensional Chern-Simons Lagrangian has to be regularised as in \eqref{reg 4dCS action}
in order to be locally integrable near each surface defect.)
After a series of technical preparations in \S\ref{sec: simple pole case} and 
\S\ref{sec: higher poles}, our main result is Theorem \ref{cor gauge invariance}, where we 
prove that the regularised $4$-dimensional Chern-Simons action
defines a gauge invariant function on the groupoid $\mathfrak{F}_{\rm bc}(X)$ 
of bulk fields $A$ and their gauge transformations $g: A\to\null^g A$,
both subject to certain boundary conditions on the surface defects, cf.\ \eqref{Fbc def}.
The boundary conditions we consider are determined by a choice of Lie subalgebra $\k \subset \g^{\widehat{\bm z}}$
of the Lie algebra $\g^{\widehat{\bm z}}$ of the product of jet groups $G^{\widehat{\bm z}} = \prod_{x\in{\bm z}} J^{n_x-1}G$,
where $n_x\geq 1$ is the order of the pole $x\in{\bm z}$ of $\omega$,
that is isotropic with respect to a non-degenerate symmetric 
invariant bilinear form $\langle\!\langle \cdot, \cdot \rangle\!\rangle_{\omega}^{}$ 
defined in terms of $\omega$. We note in passing that the appearance of jet groups has
also been observed before in examples of conformal field theories, 
see \cite{Babichenko:2012uq} and \cite{Quella:2020uhk}.

\medskip

The second goal of this paper is to clarify the passage from 
$4$-dimensional Chern-Simons theory to $2$-dimensional integrable field theories
that was proposed in \cite{Costello:2019tri}; see also \cite{Delduc:2019whp}
for some previous clarifications. The crucial new ingredient in our approach
is Theorem \ref{thm: equiv Fbc to F}, which proves that the groupoid
$\mathfrak{F}_{\rm bc}(X)$ of bulk fields with boundary conditions 
in \eqref{Fbc def} is equivalent to the groupoid $\mathfrak{F}(X)$ in \eqref{FM def}
whose objects are compatible pairs $(A,h)$ consisting of 
a bulk field $A$ and an \emph{edge mode} $h : \Sigma \to G^{\widehat{\bm z}}$ on
$\Sigma$ with values in the product of jet groups $G^{\widehat{\bm z}} = \prod_{x\in{\bm z}} J^{n_x-1}G$.
The groupoid $\mathfrak{F}(X)$ arises naturally by implementing the boundary conditions
on the surface defects by a homotopy pullback \eqref{homotopy pullback} in the model category
of groupoids, cf.\ \cite{HomEdgeModes}. Using this equivalence of groupoids, we can
transfer the regularised $4$-dimensional Chern-Simons action \eqref{reg 4dCS action}
to a gauge invariant action $S_\omega^{\rm ext}$ 
on the groupoid $\mathfrak{F}(X)$, whose explicit form \eqref{Action edge modes}
justifies the interpretation of the edge mode $h : \Sigma \to G^{\widehat{\bm z}}$
as the field content of a field theory on $\Sigma$. 

The passage to a $2$-dimensional integrable field theory consists
of finding a specific solution $A=\L$ to the bulk equation of motion determined 
by \eqref{Action edge modes} that qualifies as a Lax connection. Specifically, we introduce a
subgroupoid $\mathfrak{F}_{\rm Lax}(X)$ of $\mathfrak{F}(X)$ whose objects
are compatible pairs $(\L, h)$, where the bulk field $\L$ is meromorphic with poles
at the zeroes of $\omega$ on account of the bulk equation of motion 
and is \emph{admissible} in the sense that the defect equation of
motion can be lifted to a flatness condition for $\L$ on the whole of $X$, cf.\ \eqref{F_lax groupoid}.
We also introduce in \eqref{F_2d groupoid} a groupoid $\mathfrak{F}_{\rm 2d}(\Sigma)$ for the integrable field theory itself,
whose objects consist only of an edge mode $h : \Sigma \to G^{\widehat{\bm z}}$.
We prove in Corollary \ref{cor: projection functor equivalence} that the forgetful 
functor $\mathfrak{F}_{\rm Lax}(X) \to \mathfrak{F}_{\rm 2d}(\Sigma)$
is an equivalence of groupoids if and only if, for each $h : \Sigma \to G^{\widehat{\bm z}}$,
there exists a unique connection $\L(h)$ such that the pair $(\L(h), h)$ belongs to $\mathfrak{F}_{\rm Lax}(X)$.
In this case one is able to transfer the action on $\mathfrak{F}(X)$ all the way down
to $\mathfrak{F}_{\rm 2d}(\Sigma)$ to obtain the action for an integrable field
theory on $\Sigma$ whose Lax connection is $\L(h)$.
Unique solutions $\L$ to the compatibility condition on the pair $(\L,h)$ have 
been shown to exist in the case of single and double poles in \cite{Costello:2019tri, Delduc:2019whp}. 
We do not address the issue of solvability of this condition in the general setting of the present work.

\medskip

Let us briefly outline the content of this paper. 
In Section \ref{sec: simple pole case} we study $4$-dimensional 
Chern-Simons theory and its gauge transformations for simple poles in $\omega$. 
This is generalised in Section \ref{sec: higher poles} to the case of general poles. 
In Section \ref{sec: boundary conditions} we link gauge invariance of 
the action to suitable boundary conditions and realise that an equivalent 
description involving also edge modes can be obtained. This equivalent 
perspective is used in Section \ref{sec: passage to IFT} to explain how 
integrable field theories emerge from 4d Chern-Simons theory as particular partial solutions.

\medskip\medskip

\paragraph{\bf Notations and conventions:}

Let $G$ be a simply connected matrix Lie group over $\CC$ and let $\g$ denote its Lie algebra. 
We fix a non-degenerate invariant symmetric bilinear form $\langle \cdot, \cdot \rangle : \g \times \g \to \CC$.

Let $\omega$ be a meromorphic 1-form on $\CP$. We denote by $\bm \zeta \subset \CP$ its finite subset 
of zeroes and by $\bm z \subset \CP$ its finite subset of poles. We shall assume that $\omega$ has at 
least one zero, namely $|\bm \zeta| \geq 1$. This implies that $\omega$ has at least three poles 
(counting multiplicities) and so, in particular, $|\bm z| \geq 1$.

Let $\Sigma \coloneqq \RR^2$ and $C \coloneqq \CP \setminus \bm \zeta$. 
We consider the $4$-dimensional manifold
\begin{equation*}
X \coloneqq \Sigma \times C.
\end{equation*}
We fix a global holomorphic coordinate $z : C \to \CC$ on $C$, 
which exists because it is assumed that $|\bm \zeta| \geq 1$.
We can represent the 1-form $\omega$ in this coordinate as
\begin{equation} \label{omega higher poles pre}
\omega = \sum_{x \in \bm z} \sum_{p=0}^{n_x - 1} \frac{k^x_p \,dz}{(z - x)^{p+1}}, 
\end{equation}
where $k^x_p \in \CC$, for each $p = 0, \ldots, n_x - 1$, 
and $n_x \in \ZZ_{\geq 1}$ is the order of the pole $x \in \bm z$.
By a slight abuse of notation, we shall denote by $\omega$ also the pullback along 
the projection $p_C : X \to C$ of the restriction of $\omega$ to $C$.

Using the Cartesian product structure of $X$ and the complex structure on $C$, 
we obtain a triple grading on the vector space of differential forms 
\begin{equation}\label{triple graded forms}
\Omega^\bullet(X) = \bigoplus_{r=0}^2 \bigoplus_{s,\bar{s}=0}^{1}
\Omega^{r,s,\bar{s}}(X)
\end{equation}
and the corresponding decomposition of the de Rham differential 
as $d_X = d_\Sigma + \partial + \bar \partial$.
To simplify notation, we often denote $d_X$ simply by $d$.

\subsubsection*{Acknowledgements}
B.V.\ would like to thank S.\ Lacroix and M.\ Magro for useful discussions. 
Furthermore, we are grateful to S.\ Bunk for suggesting 
the averaging construction mentioned in Remark \ref{minimal}. 
M.B.\ gratefully acknowledges the financial support of the 
National Group of Mathematical Physics GNFM-INdAM (Italy). 
A.S.\ gratefully acknowledges the financial support of 
the Royal Society (UK) through a Royal Society University 
Research Fellowship (UF150099), a Research Grant (RG160517) 
and two Enhancement Awards (RGF\textbackslash EA\textbackslash 180270 and RGF\textbackslash EA\textbackslash 201051).

\section{Simple poles in $\omega$} \label{sec: simple pole case}
To begin with, we shall assume in this section that all the poles of $\omega$ are simple, 
i.e.\ we take $n_x = 1$ for all $x \in \bm z$. The case with higher order poles in 
$\omega$ will require a regularisation of the action, which we shall return to in \S\ref{sec: higher poles}.

\subsection{Action} \label{sec: action}
Consider the $4$-dimensional Chern-Simons action \cite{Costello:2019tri}
\begin{equation} \label{4dCS action}
S_\omega(A) = \frac{\ii}{4 \pi} \int_X \omega \wedge \CS(A),
\end{equation}
where $A\in\Omega^1(X,\g)$ is a smooth $\g$-valued 1-form on $X$ and $\CS(A) \coloneqq \langle A, dA + 
\mbox{\small $\frac{1}{3}$} [A , A] \rangle \in\Omega^3(X)$ is the Chern-Simons 3-form.

Since $\omega$ is the pullback along $p_C : X \to C$ of a meromorphic 1-form on $C$ 
with poles in $\bm z$, it is singular on the disjoint union of surface defects
\begin{equation} \label{defect simple poles}
D \coloneqq \Sigma \times \bm z = \bigsqcup_{x \in \bm z} \Sigma_x,
\end{equation}
where $\Sigma_x \coloneqq \Sigma \times \{ x \}$ for every pole $x \in \bm z$. Later we shall 
make use of the embeddings of the individual surface defects $\Sigma_x$, for each $x \in \bm z$,
and of the disjoint union $D$, which we denote respectively by
\begin{equation} \label{defect embed}
\iota_x : \Sigma_x \longhookrightarrow X, \qquad
\bm \iota : D \longhookrightarrow X.
\end{equation}

The following lemma shows that the 4-form $\omega \wedge \CS(A) \in \Omega^4(X \setminus D)$ 
is locally integrable near $D$.
\begin{lemma} \label{lem: simple poles}
For any $\eta \in \Omega^3(X)$, the $4$-form $\omega \wedge \eta \in \Omega^4(X \setminus D)$ is locally 
integrable near the surface defect $\Sigma_x$ associated with any simple pole $x \in \bm z$ of $\omega$.
\end{lemma}
\begin{proof}
We can write $\eta = \eta_{\bar z} \wedge d\bar z + \eta_z \wedge dz$, 
where $\eta_{\bar z} \in \Omega^{2, 0,0}(X)$ and $\eta_z \in \Omega^2(X)$. 
Then $\omega \wedge \eta = \omega \wedge d\bar z \wedge \eta_{\bar z}$.
Since $x$ is a simple pole of $\omega$, we can write $\omega = \frac{k^x_0}{z - x} dz + \widetilde{\omega}$,
where the meromorphic 1-form $\widetilde{\omega}$ on $C$ is regular at $x$. In terms of polar coordinates 
$z = x + r e^{\ii \theta}$ we then have $\omega \wedge d\bar z = 
- 2 \ii k^x_0 e^{-\ii \theta} dr \wedge d\theta + \widetilde{\omega} \wedge d\bar z$, 
which is locally integrable near $x$ and hence so is $\omega \wedge \eta$ near $\Sigma_x \subset X$.
\end{proof}

\subsection{Gauge transformations} \label{sec: gauge inv action}
Consider the left action of the group $C^\infty(X, G)$ on $\Omega^1(X, \g)$ defined by
\begin{align} \label{gauge transf}
C^\infty(X, G) \times \Omega^1(X, \g) &\longrightarrow \Omega^1(X, \g),\\
(g, A) &\longmapsto \null^g A \coloneqq - d g g^{-1} + g A g^{-1}. \notag
\end{align}
Under a gauge transformation $g : A \rightarrow \null^g A$, the action \eqref{4dCS action} transforms as
\begin{equation} \label{S4d gauge transf}
S_\omega(\null^g A) = S_\omega(A) + \frac{\ii}{4 \pi} \int_X \omega \wedge d \langle g^{-1} d g , A \rangle + \frac{\ii}{4 \pi} \int_X \omega \wedge g^\ast \chi_G,
\end{equation}
where $\chi_G \coloneqq \frac{1}{6} \langle \theta_G, [\theta_G , \theta_G] \rangle\in \Omega^3(G)$ is 
the Cartan 3-form on $G$ and $\theta_G\in \Omega^1(G,\g)$ denotes the left Maurer-Cartan form on $G$, so that 
$g^\ast \chi_G = \frac{1}{6} \langle g^{-1} d g, [g^{-1} d g , g^{-1} d g] \rangle$. 

\medskip

Define the \emph{defect group} $G^{\bm z}$ and its Lie algebra $\g^{\bm z}$ as
\begin{equation*}
G^{\bm z} \coloneqq \prod_{x \in \bm z} G, \qquad
\g^{\bm z} \coloneqq \prod_{x \in \bm z} \g.
\end{equation*}
We endow $\g^{\bm z}$ with the non-degenerate invariant symmetric bilinear form
\begin{equation} \label{bilinear form gD}
\langle\!\langle \cdot, \cdot \rangle\!\rangle_{\omega}^{} : \g^{\bm z} \times \g^{\bm z} \longrightarrow \CC, \qquad
\langle\!\langle X, Y \rangle\!\rangle_{\omega}^{} \coloneqq \sum_{x \in \bm z} k^x_0 \, \langle X_x, Y_x \rangle,
\end{equation}
for every $X = (X_x)_{x \in \bm z},~ Y = (Y_x)_{x \in \bm z} \in \g^{\bm z}$,
where $k^x_0\in\CC$ is the residue of $\omega$ at $x \in \bm z$.
For $\g$-valued 1-forms on $D$ and smooth $G$-valued maps on $D$, we have the isomorphisms
\begin{subequations} 
\begin{align}
 \Omega^1(D, \g) &\cong \prod_{x \in \bm z} \Omega^1(\Sigma_x, \g) \cong \Omega^1(\Sigma, \g^{\bm z}),\\
 C^\infty(D, G) &\cong \prod_{x \in \bm z} C^\infty(\Sigma_x, G) \cong C^\infty(\Sigma, G^{\bm z}).
\end{align}
\end{subequations}
The pullbacks by the second embedding in \eqref{defect embed} of $\g$-valued $1$-forms on $X$ 
and of smooth $G$-valued maps on $X$ can therefore be thought of as maps
\begin{equation*}
\bm \iota^\ast : \Omega^1(X, \g) \longrightarrow \Omega^1(\Sigma, \g^{\bm z}), \qquad
\bm \iota^\ast : C^\infty(X, G) \longrightarrow C^\infty(\Sigma, G^{\bm z}).
\end{equation*}

\begin{lemma} \label{lem: CP formula}
For any $\eta \in \Omega^2(X)$, we have
\begin{equation*}
\int_X \omega \wedge d \eta = 2 \pi \ii \sum_{x \in \bm z} k^x_0 \int_{\Sigma_x} \iota_x^\ast \eta.
\end{equation*}
\end{lemma}
\begin{proof}
Recalling our notations and conventions at the end 
of Section \ref{sec: intro}, we have
\begin{equation*}
\int_X \omega \wedge d \eta = 
\int_X \omega\wedge (d_\Sigma + \bar\partial) \eta =\int_X \omega \wedge \bar\partial \eta  - \int_X d_\Sigma (\omega \wedge \eta) ,
\end{equation*}
where we used the decomposition of the de Rham differential
$d \eta = d_\Sigma \eta + \partial \eta + \bar\partial \eta$ 
and the fact that $\omega$ is the pullback along
$p_C : X \to C$ of a meromorphic 1-form on $C$, hence $\omega \wedge \partial \eta = 0$,
and $d_\Sigma (\omega \wedge \eta) = - \omega \wedge d_\Sigma \eta$.
The second term in the equation displayed above vanishes by Stokes' theorem on $\Sigma$. The result now follows by the Cauchy-Pompeiu integral formula.
\end{proof}

\begin{proposition} \label{prop: kinetic term}
For any $g \in C^\infty(X, G)$ and $A \in \Omega^1(X, \g)$, we have
\begin{equation*}
\int_X \omega \wedge d \langle g^{-1} d g , A \rangle = 2 \pi \ii \int_\Sigma \big\langle{\mkern-5mu}\big\langle (\bm \iota^\ast g)^{-1} d_\Sigma (\bm \iota^\ast g) , \bm \iota^\ast A \big\rangle{\mkern-5mu}\big\rangle_{\omega}^{}.
\end{equation*}
\end{proposition}
\begin{proof}
Applying Lemma \ref{lem: CP formula}, we obtain
\begin{equation*}
\int_X \omega \wedge d \langle g^{-1} d g , A \rangle 
= 2 \pi \ii \sum_{x \in \bm z} \int_{\Sigma_x} k^x_0 \, \big\langle (\iota_x^\ast g)^{-1} d_{\Sigma_x} (\iota_x^\ast g) , \iota_x^\ast A \big\rangle.
\end{equation*}
The result follows by definition \eqref{bilinear form gD} of the bilinear form on $\g^{\bm z}$.
\end{proof}

By Proposition \ref{prop: kinetic term}, the second term on the right hand side of \eqref{S4d gauge transf} 
now manifestly depends only on the defect fields $\bm \iota^\ast g \in C^\infty(\Sigma, G^{\bm z})\cong C^\infty(D,G)$ 
and $\bm \iota^\ast A \in \Omega^1(\Sigma, \g^{\bm z})\cong \Omega^1(D,\g)$. We will show in Proposition 
\ref{prop: WZ depends on D} below that the same is true for the third term on the right hand side 
of \eqref{S4d gauge transf}. To prove this, we first need to introduce further notations and techniques.

\medskip

For a manifold $M$ and a closed subset $S \subset M$ with embedding $\iota : S \hookrightarrow M$, 
let $C^0_S(M, G)$ (resp.\ $C^\infty_S(M, G)$) denote the set of continuous (resp.\ smooth) 
maps $g : M \to G$ such that $\iota^\ast g = e$, where by abuse of notation $e$ denotes the constant 
map $S \to G$ to the identity element $e \in G$.

Let $I \coloneqq [0, 1] \subset \RR$ denote the closed unit interval and define the maps 
$j_t : M \hookrightarrow M \times I,~ p \mapsto (p, t)$, for every $t \in I$. 
A \emph{relative continuous} (resp.\ \emph{smooth}) \emph{homotopy} between two maps $g, g' \in C^0_S(M, G)$ 
(resp.\ $g, g' \in C^\infty_S(M, G)$) is a map $H \in C^0_{S \times I}(M \times I, G)$ (resp.\ $H \in C^\infty_{S \times I}(M \times I, G)$) 
such that
\begin{equation*}
j_0^\ast H = g, \qquad j_1^\ast H = g'.
\end{equation*}
We write $g \sim_S^{} g'$ (resp.\ $g \sim^\infty_S g'$) and say that $g$ and $g'$ are 
homotopic relative to $S$. This defines equivalence relations $\sim_S^{}$ on 
$C^0_S(M, G)$ and $\sim_S^\infty$ on $C^\infty_S(M, G)$.

\begin{lemma} \label{lem: smooth vs continuous}
The canonical map
\begin{equation*}
C^\infty_D(X, G) \big/\!\! \sim^\infty_D \;\; \longrightarrow \;\; C^0_D(X, G) \big/\!\! \sim_D^{}
\end{equation*}
is a bijection.
\end{lemma}
\begin{proof}
Let $g, g' \in C^\infty_D(X, G)$ be such that $g \sim_D^{} g'$. By \cite[Theorem 6.29]{LeeBook},
it follows that $g \sim^\infty_D g'$. Hence, the given map is injective.

Now let $g \in C^0_D(X, G)$. Then $\bm \iota^\ast g = e$ is smooth,
so by \cite[Theorem 6.26]{LeeBook} it follows that $g \sim_D^{} g'$ for some $g' \in C^\infty_D(X, G)$. 
Hence, the given map is surjective.
\end{proof}

Recall the projection $p_C : X \to C$. For any $a \in \Sigma$, we also consider the smooth embedding
$i_a : C \hookrightarrow X,~z \mapsto (a, z)$. We have that $p_C(D) = \bm z$ and $i_a(\bm z) \subset D$.

\begin{lemma} \label{lem: contracting R2}
For any $a \in \Sigma$, the maps
\begin{equation*}
\begin{tikzcd}
C^0_D(X, G) \big/ \!\!\sim_D^{} \arrow[r, shift left = 1, "i_a^\ast"] & C^0_{\bm z}(C, G) \big/ \!\!\sim_{\bm z}^{} \arrow[l, shift left = 1, "p_C^\ast"]
\end{tikzcd}
\end{equation*}
exhibit a bijection.
\end{lemma}
\begin{remark}
The maps $p_C^\ast$ and $\iota_a^\ast$ are well-defined. Indeed, suppose more generally that 
$M$, $M'$ are topological spaces with closed subsets $S \subset M$, $S' \subset M'$ and corresponding embedding 
maps $\iota : S \hookrightarrow M$ and $\iota' : S' \hookrightarrow M'$. Let $f : M \to M'$ be a continuous 
map such that $f(S) \subset S'$. Then the pullback by $f$ induces a map $f^\ast : C^0_{S'}(M', G) \to C^0_{S}(M, G)$. 
Indeed, if $g \in C^0_{S'}(M', G)$ then $f^\ast g \in C^0_S(M, G)$ since
\begin{equation*}
\iota^\ast f^\ast g = (f \circ \iota)^\ast g = (\iota' \circ f|_S)^\ast g = f|_S^\ast \, \iota'^\ast g = e,
\end{equation*}
where $f|_S : S \to S'$ is the restriction of $f$ to $S \subset M$ and in the final step we 
used the fact that $\iota'^\ast g = e$. Moreover, given any relative homotopy 
$H \in C^0_{S' \times I}(M' \times I, G)$, 
we have $(f \times \id)^\ast H \in C^0_{S \times I}(M \times I, G)$ since $(\iota \times \id)^\ast (f \times \id)^\ast H = e$ 
by the same computation as above. We therefore obtain a well-defined map between the relative 
homotopy classes $f^\ast : C^0_{S'}(M', G) \big/\!\!\sim_{S'}^{} \;\to C^0_{S}(M, G)\big /\!\!\sim_S^{}$, as required.
\end{remark}
\begin{proof}[Proof of Lemma \ref{lem: contracting R2}]
Let $a \in \Sigma$. We have to show that $i_a^\ast$ and $p_C^\ast$ are inverses of each other.
Since $p_C \circ i_a = \id_C$, we have $i_a^\ast\, p_C^\ast = (p_C \circ i_a)^\ast = \id$.

Consider now $i_a \circ p_C : X \to X$. We have a continuous homotopy
\begin{equation*}
H : X \times I \longrightarrow X, \qquad
(p, z, t) \longmapsto \big( (1-t)p + t a, z \big)
\end{equation*}
between $\id_X$ and $i_a \circ p_C$.
Note that $H(D \times I) \subset D$, in other words $H \circ (\bm \iota \times \id) = \bm \iota \circ H|_{D \times I}$.
For any $g \in C^0_D(X, G)$, the continuous map
$g \circ H : X \times I \to G$
belongs to $C^0_{D \times I}(X \times I, G)$ since
\begin{equation*}
(\bm \iota \times \id)^\ast (g \circ H) = g \circ H \circ (\bm \iota \times \id) = 
g \circ \bm \iota \circ H|_{D \times I} = (\bm \iota^\ast g) \circ H|_{D \times I} = e.
\end{equation*}
In the final equality we used the fact that $\bm \iota^\ast g = e$ since $g \in C^0_D(X, G)$. 
Moreover, $j_0^\ast (g \circ H) = g$ and $j_1^\ast (g \circ H) = g \circ i_a \circ p_C = p_C^\ast i_a^\ast g$ 
so that $g \circ H$ is a relative continuous homotopy between $g$ and $p_C^\ast i_a^\ast g$, 
i.e.\ $p_C^\ast i_a^\ast g \sim_D^{} g$.
Hence $p_C^\ast i_a^\ast = \id$, as required.
\end{proof}

\begin{lemma} \label{lem: singleton}
$C^0_{\bm z}(C, G) \big/ \!\!\sim_{\bm z}^{}$ is a singleton.
\end{lemma}
\begin{proof}
A relative continuous homotopy $H \in C^0_{\bm z \times I}(C \times I, G)$ between 
two maps $g, g' \in C^0_{\bm z}(C, G)$ is a continuous path in the mapping space
$\Map_{\bm z}(C, G)$ from $g$ to $g'$. Thus
\begin{equation*}
C^0_{\bm z}(C, G) \big/ \!\!\sim_{\bm z}^{} \; \cong \; \pi_0 \big( \Map_{\bm z}(C, G) \big).
\end{equation*}
Now fix any point $x \in \bm z$. The inclusion $i : \bm z \hookrightarrow C$ induces a continuous map
\begin{equation*}
i^\ast : \Map_{\{ x \}}(C, G) \longrightarrow \Map_{\{ x \}}(\bm z, G)
\end{equation*}
between based mapping spaces,
whose fibre over the constant map $e \in \Map_{\{ x \}}(\bm z, G)$ is $\Map_{\bm z}(C, G)$. Hence, we get a fibre sequence
\begin{equation*}
\Map_{\bm z}(C, G) \longhookrightarrow \Map_{\{ x \}}(C, G) \overset{i^\ast}\longrightarrow \Map_{\{ x \}}(\bm z, G).
\end{equation*}
Since $i : \bm z \hookrightarrow C$ is a cofibration, it follows that $i^\ast$ is a fibration 
and hence we obtain a long exact sequence of homotopy groups
\begin{equation*}
\ldots \longrightarrow \pi_1\big( \Map_{\{ x \}}(\bm z, G) \big) \longrightarrow \pi_0\big( \Map_{\bm z}(C, G) \big) \longrightarrow \pi_0\big( \Map_{\{ x \}}(C, G) \big) \longrightarrow \ldots
\end{equation*}
Observe that $\pi_1(\Map_{\{ x \}}(\bm z, G)) \cong \pi_1(G)^{|\bm z| - 1}\cong \{\ast\}$ is trivial 
since $G$ is assumed to be simply connected. To compute $\pi_0( \Map_{\{ x \}}(C, G))$,
we recall that $C=\CP\setminus {\bm \zeta}$ is topologically a 2-sphere $S^2$ with $|{\bm \zeta}|\geq 1$ punctures. 
Hence, there exists a deformation retract from $C$
to a bouquet of circles $\bigvee^{|{\bm \zeta}|-1} S^1 $, where $\vee$ denotes
the wedge sum (i.e.\ categorical coproduct) of pointed topological spaces.
It then follows that  
\begin{equation*}
\pi_0\big( \Map_{\{ x \}}(C, G)\big) \cong \pi_0 \big(\Map_{\{x\}}(S^1,G)\big)^{|{\bm \zeta}|-1} 
\cong \pi_1(G)^{|\bm\zeta| -1} \cong \{\ast\}
\end{equation*}
is trivial since $G$ is assumed to be simply connected. 
From the long exact sequence we conclude that 
$\pi_0(\Map_{\bm z}(C, G))$ is a singleton, which completes the proof.
\end{proof}

\begin{proposition} \label{prop: WZ depends on D}
The integral $\int_X \omega \wedge g^\ast \chi_G$ depends on $g \in C^\infty(X, G)$ 
only through $\bm \iota^\ast g \in C^\infty(\Sigma, G^{\bm z})$.
\end{proposition}
\begin{remark}
The present situation is to be contrasted with the usual WZ-term in the WZW model action. 
Indeed, to even write the latter down one has to extend a field $g \in C^\infty(S^2, G)$ to a field 
$\widetilde{g} \in C^\infty(B^3, G)$ on the $3$-dimensional ball $B^3$ with $\partial B^3 = S^2$. This is possible 
as $\pi_2(G) = 0$ but the extension $\widetilde{g}$ is not unique. The set of homotopy classes of smooth maps 
$\widetilde{g} \in C^\infty(B^3, G)$ with $\widetilde{g}|_{S^2} = g$ is measured by $\pi_3(G)$, which for a 
simple Lie group $G$ is given by $\pi_3(G) \cong \ZZ$. For the extensions $\widetilde{g}$ in different 
homotopy classes, the integrals $\int_{B^3} \widetilde{g}^\ast \chi_G$ differ by integer multiples of a constant.
\end{remark}
\begin{proof}[Proof of Proposition \ref{prop: WZ depends on D}]
For any $g, h \in C^\infty(X, G)$, we have the Polyakov-Wiegmann identity
\begin{equation*}
(gh^{-1})^\ast \chi_G 
= g^\ast \chi_G - h^\ast \chi_G + d \langle g^{-1} d g, h^{-1} dh \rangle.
\end{equation*}
Using Lemma \ref{lem: CP formula} and the definition of the bilinear form \eqref{bilinear form gD} on $\g^{\bm z}$, we find
\begin{align*}
\int_X \omega \wedge d\langle g^{-1} d g, h^{-1} dh \rangle &= 2 \pi \ii \sum_{x \in \bm z} \int_{\Sigma_x} k^x_0 \,\big\langle (\iota_x^\ast g)^{-1} d_{\Sigma_x}(\iota_x^\ast g), (\iota_x^\ast h)^{-1} d_{\Sigma_x}(\iota_x^\ast h) \big\rangle\\
&= 2 \pi \ii \int_\Sigma \big\langle{\mkern-5mu}\big\langle (\bm \iota^\ast g)^{-1} d_\Sigma(\bm \iota^\ast g), (\bm \iota^\ast h)^{-1} d_\Sigma(\bm \iota^\ast h) \big\rangle{\mkern-5mu}\big\rangle_{\omega}^{}.
\end{align*}
In particular, if $\bm \iota^\ast g = \bm \iota^\ast h$ then this vanishes by the skew-symmetry of the bilinear pairing 
$\langle\!\langle  \cdot, \cdot \rangle\!\rangle_{\omega}^{} : 
\Omega^1(\Sigma, \g^{\bm z}) \times \Omega^1(\Sigma, \g^{\bm z}) \to \Omega^2(\Sigma)$ on 1-forms. 
It follows that
\begin{equation} \label{PW identity}
\int_X \omega \wedge (gh^{-1})^\ast \chi_G = \int_X \omega \wedge g^\ast \chi_G - \int_X \omega \wedge h^\ast \chi_G
\end{equation}
for any $g, h \in C^\infty(X, G)$ such that $\bm \iota^\ast g = \bm \iota^\ast h$.

\medskip

The latter condition can be equivalently stated as $\bm \iota^\ast (g h^{-1}) = e$, or in other words 
$g h^{-1} \in C^\infty_D(X, G)$. By Lemmas \ref{lem: smooth vs continuous}, \ref{lem: contracting R2} and \ref{lem: singleton} 
we deduce that $C^\infty_D(X, G) \big/ \!\! \sim_D^\infty$ is a singleton.
Hence, there exists a relative smooth homotopy $H \in C^\infty_{D \times I}(X \times I, G)$
between $gh^{-1}$ and $e \in C^\infty_D(X, G)$, i.e.\ $j_0^\ast H = g h^{-1}$ and $j_1^\ast H = e$.
Then
\begin{align} \label{WZ g h inv exact}
d\bigg( \int_{I} H^\ast \chi_G \bigg) &= \int_{I} d H^\ast \chi_G = \int_{I} \big( d_{X \times I} - d_{I} \big) H^\ast \chi_G\\
&= - \int_{I} d_{I} H^\ast \chi_G = j_0^\ast H^\ast \chi_G - j_1^\ast H^\ast \chi_G
= (g h^{-1})^\ast \chi_G, \notag
\end{align}
where in the second step $d_{X \times I} = d + d_{I}$ is the differential on $\Omega^\bullet(X \times I)$.
In the third equality we used the fact that $H^\ast \chi_G \in \Omega^3(X \times I)$ is closed, 
i.e.\ $d_{X \times I} H^\ast \chi_G = 0$, and in the second last step we used Stokes' theorem. 
In the final step we used the fact that $e^\ast \chi_G = 0$.

Taking the wedge product of \eqref{WZ g h inv exact} with $\omega$ and integrating over $X$ we obtain, 
using again Lemma \ref{lem: CP formula},
\begin{align} \label{omega WZ gh inv}
\int_X \omega \wedge (gh^{-1})^\ast \chi_G &= \int_X \omega \wedge d\bigg( \int_{I} H^\ast \chi_G \bigg)\\ 
&= 2 \pi \ii \sum_{x \in \bm z} k^x_0 \int_{\Sigma_x \times I} (\iota_x \times \id)^\ast H^\ast \chi_G = 0. \notag
\end{align}
In the last equality we used the fact that $(\iota_x \times \id)^\ast H = e \in C^\infty(\Sigma_x \times I, G)$,
for every $x \in \bm z$, and again that $e^\ast \chi_G = 0$. Finally, by combining \eqref{omega WZ gh inv} 
with \eqref{PW identity}, it follows that $\int_X \omega \wedge g^\ast \chi_G = \int_X \omega \wedge h^\ast \chi_G$
for any $g, h \in C^\infty(X, G)$ such that $\bm \iota^\ast g = \bm \iota^\ast h$. This completes the proof.
\end{proof}

Recall the bilinear form $\langle\!\langle \cdot, \cdot \rangle\!\rangle_{\omega}^{}$ on the Lie algebra $\g^{\bm z}$ 
introduced in \eqref{bilinear form gD}. We let $\chi_{G^{\bm z}} \coloneqq \frac{1}{6} 
\langle\!\langle \theta_{G^{\bm z}}, [ \theta_{G^{\bm z}} , \theta_{G^{\bm z}}] \rangle\!\rangle_{\omega}^{} \in\Omega^3(G^{\bm z})$ 
denote the corresponding Cartan 3-form on $G^{\bm z}$, where $\theta_{G^{\bm z}}\in \Omega^1(G^{\bm z},\g^{\bm z})$ 
is the left Maurer-Cartan form on $G^{\bm z}$. 
Since $\Omega^1(G^{\bm z}, \g^{\bm z}) \cong \prod_{x \in \bm z} \Omega^1(G^{\bm z}, \g)$, 
we can express $\theta_{G^{\bm z}}$ as a tuple $(\theta^x_G)_{x \in \bm z}$ 
of $\g$-valued 1-forms on $G^{\bm z}$. Here, for each $x \in \bm z$, 
$\theta^x_G = \pi_x^\ast \theta_G \in \Omega^1(G^{\bm z}, \g)$ is the pullback 
of the left Maurer-Cartan form $\theta_G$ on $G$ along the canonical 
projection $\pi_x : G^{\bm z} \to G$ onto the $x$-factor of $G^{\bm z}$. 
It then also follows that $\chi_{G^{\bm z}} = \sum_{x \in \bm z} k^x_0 \, \chi^x_G$, where 
$\chi^x_G \coloneqq \frac{1}{6} \langle \theta^x_G, [\theta^x_G, \theta^x_G] \rangle = \pi_x^\ast \chi_G \in \Omega^3(G^{\bm z})$.

\begin{proposition} \label{prop: int WZ compute}
For any $g \in C^\infty(X, G)$, we have
\begin{equation*}
\int_X \omega \wedge g^\ast \chi_G = - 2 \pi \ii \int_{\Sigma \times I} \widehat{g}^\ast \chi_{G^{\bm z}} ,
\end{equation*}
where $\widehat{g} \in C^\infty(\Sigma \times I, G^{\bm z})$ is any lazy homotopy between 
$\bm \iota^\ast g \in C^\infty(\Sigma, G^{\bm z})$ and the constant map $e \in C^\infty(\Sigma, G^{\bm z})$. 
\end{proposition}
\begin{proof}
First we note that since $\Sigma = \RR^2$ is contractible and $G^{\bm z}$ is connected, as $G$ is, 
there exists a lazy homotopy $\widehat{g} \in C^\infty(\Sigma \times I, G^{\bm z})$ between 
$\bm \iota^\ast g$ and $e$, namely such that $\widehat{g}(\--, t) = \bm \iota^\ast g$ for 
$t$ near $0$ and $\widehat{g}(\--, t) = e$ for $t$ near $1$. 

Let us denote by $\Delta$ the unit disc and by $\varrho : \Delta \to I$ the radial coordinate. 
Let $\Delta_x \subset C$ be disjoint discs around each $x \in \bm z$. We then have the following isomorphism
\begin{equation*}
C^\infty\bigg( \bigsqcup_{x \in \bm z} \Sigma \times \Delta_x, G \bigg) \cong \prod_{x \in \bm z} C^\infty(\Sigma \times \Delta_x, G) \cong C^\infty(\Sigma \times \Delta, G^{\bm z}).
\end{equation*}
Consider $(\id_\Sigma \times \varrho)^\ast \widehat{g} \in C^\infty(\Sigma \times \Delta, G^{\bm z})$, 
regard it as a smooth map $\bigsqcup_{x \in \bm z} \Sigma \times \Delta_x \to G$ 
under the above isomorphism and extend the latter to the whole of $X$ by the identity $e \in G$. 
By construction, this defines a smooth map $\widetilde{g} \in C^\infty(X, G)$ such that 
$\bm \iota^\ast \widetilde{g} = \bm \iota^\ast g \in C^\infty(\Sigma, G^{\bm z})$. 
(Note that $\widetilde{g}$ is smooth because $\widehat{g}$ is lazy.)
By Proposition \ref{prop: WZ depends on D}, we deduce
\begin{equation*}
\int_X \omega \wedge g^\ast \chi_G = \int_X \omega \wedge \widetilde{g}^\ast \chi_G.
\end{equation*}
It remains to compute the integral on the right hand side. This can be done directly as in \cite[\S3.3]{Delduc:2019whp}. 
In view of generalising this computation to the higher order pole case later in \S\ref{sec: higher poles}, 
it is useful to repeat it in the present language. We find 
\begin{align*}
&\int_X \omega \wedge \widetilde{g}^\ast \chi_G = \sum_{x \in \bm z} \int_{\Sigma \times \Delta_x} \omega \wedge \widetilde{g}^\ast \chi_G = \sum_{x \in \bm z} \int_{\Sigma \times \Delta_x} \omega \wedge d \bigg( -\int_{\gamma_z} \widetilde{g}^\ast \chi_G \bigg)\\
&\quad = -2 \pi \ii \sum_{x \in \bm z} k^x_0 \int_{\Sigma_x} \int_{\gamma_x} \widetilde{g}^\ast \chi_G
= - 2 \pi \ii \sum_{x \in \bm z} k^x_0 \int_{\Sigma \times I} \widehat{g}^\ast \chi_G^x
= - 2 \pi \ii \int_{\Sigma \times I} \widehat{g}^\ast \chi_{G^{\bm z}}. 
\end{align*}
The first equality follows from noting that $\widetilde{g}^\ast \chi_G$ vanishes outside of 
$\bigsqcup_{x \in \bm z} \Sigma \times \Delta_x \subset X$. In the second step, we used the fact that 
$\widetilde{g}^\ast \chi_G$ is closed, hence exact on the contractible 
subspaces $\Sigma \times \Delta_x\subset X$. In particular, the value of an explicit primitive 
$-\int_{\gamma_z} \widetilde{g}^\ast \chi_G$ at the point $(p, z) \in \Sigma \times \Delta_x$ 
is given by the integral along a radial path $\gamma_z:I \to \Delta_x$ from $(p, z)$ to a point 
$(p,z_0)$ lying on the boundary of $\Sigma \times \Delta_x$. 
In the third equality we used Lemma \ref{lem: CP formula}. In the second last step we 
used the identification of $\Sigma_x \times \gamma_x(I)$ with $\Sigma \times I$
and that of $\widetilde{g} : \Sigma_x\times \gamma_x(I)\to G$ with $\pi_x \circ \widehat{g} : \Sigma\times I \to G$. 
The last equality follows from the identity $\chi_{G^{\bm z}} = \sum_{x \in \bm z} k^x_0 \,\chi_G^x$.
\end{proof}

\section{Higher order poles in $\omega$} \label{sec: higher poles}
We would like to extend the constructions of \S\ref{sec: simple pole case} to the case 
when the meromorphic $1$-form $\omega$ has higher order poles.
The immediate problem we face is that $\omega \wedge \CS(A)$ is not locally integrable 
around such a higher order pole $x$. We will therefore begin by introducing a regularisation 
of the action \eqref{4dCS action}. 
A closely related approach to the 
regularisation of integrals on Riemann surfaces appeared in \cite{LZ20} 
shortly after the first version of this paper became available.

\subsection{Regularised action}
Let $n \coloneqq \max\; \{ n_x \}_{x \in \bm z}$ denote the maximal order among all the 
poles of $\omega$. Consider the Weil algebra $\mathcal T^n \coloneqq \CC[\varepsilon] / (\varepsilon^n)$ 
of order $n$. (If $n_x = 1$ for all $x \in \bm z$ then $n=1$ and hence $\T^n \cong \CC$.) 

For each $\mathcal{T}^n$-valued $r$-form
$\zeta = \sum_{p=0}^{n-1} \zeta_p \otimes \varepsilon^p \in \Omega^r(X) \otimes_\CC \T^n$, 
we define the \emph{regularised wedge product} with $\omega$ (cf.\ \eqref{omega higher poles pre}) as
\begin{equation} \label{reg def}
(\omega \wedge \zeta)_{\rm reg}^{} \coloneqq \sum_{x \in \bm z} \sum_{p=0}^{n_x - 1} \frac{k^x_p \,dz}{z - x} \wedge \zeta_p \in \Omega^{r+1}(X \setminus D),
\end{equation}
where $D=\bigsqcup_{x\in{\bm z}}\Sigma_x\subset X$ is defined in \eqref{defect simple poles} as the disjoint union of the surface defects 
$\Sigma_x = \Sigma \times \{ x \}$. As a consequence of Lemma \ref{lem: simple poles}, 
we obtain
\begin{corollary}\label{cor: reg integrable}
For any $\zeta \in \Omega^3(X) \otimes_\CC \T^n$, the $4$-form $(\omega \wedge \zeta)_{\rm reg}^{} \in \Omega^4(X \setminus D)$ 
is locally integrable near $D$. 
\end{corollary}

We have the morphism of vector spaces (or $C^\infty$-rings in the case $r=0$)
\begin{equation} \label{jet prolongation}
j_X^\ast : \Omega^r(X) \longrightarrow \Omega^r(X)\otimes_\CC \mathcal T^n, \qquad
\eta \longmapsto \sum_{p=0}^{n-1} \frac{1}{p!} \partial_z^p \eta \otimes \varepsilon^p,
\end{equation}
given by the holomorphic part of the $(n - 1)$-jet prolongation of smooth $r$-forms on 
$X$, for any $r = 0, \ldots, 4$.
The regularised wedge product \eqref{reg def} can be related 
as follows to the ordinary wedge product.
\begin{lemma} \label{lem: CSA multiple pole}
For any $\eta \in \Omega^3(X)$, we have a decomposition
\begin{equation*}
\omega \wedge \eta = (\omega \wedge j_X^\ast \eta)_{\rm reg}^{} + d\psi,
\end{equation*}
where $\psi \in \Omega^3(X \setminus D)$ is singular on $\Sigma_x$ 
for $x \in \bm z$ if $n_x > 1$ and $\psi = 0$ if $n = 1$.
\end{lemma}
\begin{proof}
We can rewrite \eqref{omega higher poles pre} as
\begin{equation*}
\omega = \sum_{x \in \bm z} \sum_{p=0}^{n_x - 1} \frac{(-1)^p}{p!} \partial_z^p \bigg( \frac{k^x_p}{z-x} \bigg) dz.
\end{equation*}
Taking the wedge product with $\eta$, it then follows from the Leibniz rule that
\begin{equation*}
\omega \wedge \eta = \sum_{x \in \bm z} \sum_{p=0}^{n_x - 1} \sum_{r=0}^p \frac{(-1)^{p-r}}{r! (p-r)!} dz \wedge \partial_z^{p-r} \bigg( \frac{k^x_p}{z-x} \partial_z^r \eta \bigg).
\end{equation*}
The terms with $r=p$ yield $(\omega \wedge j_X^\ast \eta)_{\rm reg}^{}$. All of the remaining terms in the
sum over $r$ can be written as the de Rham differential of
\begin{equation*}
\psi = \sum_{x \in \bm z} \sum_{p=0}^{n_x - 1} \sum_{r=0}^{p-1} \frac{(-1)^{p-r}}{r! (p-r)!} \partial_z^{p-1-r} \bigg( \frac{k^x_p}{z-x} \partial_z^r \eta^{2,0,1} \bigg),
\end{equation*}
where $\eta^{2,0,1} \in \Omega^{2,0,1}(X)$ denotes the $(2,0,1)$-component 
of $\eta \in \Omega^{3}(X)$ with respect to the grading in \eqref{triple graded forms}.
This $\psi$ is singular on $\Sigma_x$ if $n_x > 1$ and vanishes if $n_x = 1$ for all $x \in \bm z$.
\end{proof}

In the case when $\omega$ has higher order poles, Lemma \ref{lem: CSA multiple pole} 
and Corollary \ref{cor: reg integrable} motivate the following definition of the \emph{regularised action}
\begin{equation} \label{reg 4dCS action}
S_\omega(A) \coloneqq \frac{\ii}{4 \pi} \int_X \big( \omega \wedge j_X^\ast \CS(A) \big)_{\rm reg}^{}.
\end{equation}
This reduces to the action \eqref{4dCS action} of \cite{Costello:2019tri} in the case when $\omega$ only has simple poles.

\subsection{Gauge transformations}
Under a gauge transformation $g: A \rightarrow \null^g A$ as in \eqref{gauge transf}, 
the regularised action \eqref{reg 4dCS action} transforms as (cf.\ \eqref{S4d gauge transf})
\begin{align} \label{reg S4d gauge transf}
S_\omega(\null^g A) = S_\omega(A) + \frac{\ii}{4 \pi} \int_X \big( \omega \wedge j_X^\ast d \langle g^{-1} d g , A \rangle \big)_{\rm reg}^{}
+ \frac{\ii}{4 \pi} \int_X \big( \omega \wedge j_X^\ast (g^\ast \chi_G) \big)_{\rm reg}^{}, 
\end{align}
where the Cartan 3-form $\chi_G\in \Omega^3(G)$ on $G$ was defined in \S\ref{sec: gauge inv action}.

Consider the Weil algebra $\mathcal T^{n_x}_x \coloneqq \CC[\varepsilon_x] / (\varepsilon_x^{n_x})$ 
of order $n_x$, the order of the pole $x\in \bm z$ of $\omega$. (Note that for a simple pole $n_x =1$ and thus
$\mathcal T^{n_x}_x \cong \CC$.) We denote by $\ell \mathcal T^{n_x}_x$ the locus of the Weil algebra,
which is a formal manifold in the context of synthetic differential geometry \cite{KockBook}.
Loosely speaking, one should think of $\ell \mathcal T^{n_x}_x$ as an infinitesimal thickening
of the point $x\in {\bm z}$.
In the present setting, the surface defects $\Sigma_x$ of \S\ref{sec: action} are replaced 
by formal manifolds 
\begin{equation*}
\Sigma_x^{n_x} \coloneqq \Sigma \times \ell \mathcal T^{n_x}_x,
\end{equation*}
for each $x \in \bm z$. 
The disjoint union of the surface defects $\Sigma_x$ in \eqref{defect simple poles} is then replaced 
by the disjoint union of their infinitesimal thickenings $\Sigma_x^{n_x}$, namely
\begin{equation*}
\widehat{D} \coloneqq  \bigsqcup_{x \in \bm z} \Sigma_x^{n_x} .
\end{equation*}
For each $x \in \bm z$, we have a morphism of $C^\infty$-rings
\begin{subequations}
\begin{equation} \label{morph formal man}
j_x^\ast : C^\infty(X) \longrightarrow C^\infty(\Sigma_x )\otimes_\CC  \mathcal T^{n_x}_x, \qquad
f\longmapsto \sum_{p=0}^{n_x-1} \frac{1}{p!}  \iota_x^\ast (\partial_z^p f) \otimes \varepsilon_x^p,
\end{equation}
given by pulling back along $\iota_x : \Sigma_x \to X$  
the holomorphic part of the $(n_x - 1)$-jet prolongation of the smooth function $f$. 
It defines a morphism of formal manifolds
\begin{equation} \label{defect embed higher}
j_x : \Sigma^{n_x}_x \longrightarrow X.
\end{equation}
\end{subequations}
The canonical induced morphism 
\begin{subequations}
\begin{equation}
\bm j^\ast : C^\infty(X) \to  \prod_{x \in \bm z} C^\infty(\Sigma_x)\otimes_\CC \mathcal T^{n_x}_x
\end{equation}
to the product of $C^\infty$-rings defines a morphism of formal manifolds
\begin{equation} \label{D embed higher}
\bm j : \widehat{D} \longrightarrow X.
\end{equation}
\end{subequations}
The pair of morphisms \eqref{defect embed higher} and \eqref{D embed higher} play an 
analogous role to the embeddings \eqref{defect embed} in the higher pole case.

We generalise the definition of the defect group $G^{\bm z}$ and its Lie algebra $\g^{\bm z}$ 
from \S\ref{sec: gauge inv action} to the case of 
higher order poles as follows. Recall that, for each $k\geq 1$, the mapping space 
$C^\infty(\ell \mathcal{T}^k, M)$ from $\ell \mathcal{T}^k$ to a manifold $M$ is a manifold, 
namely the total space of the bundle of $(k-1)$-jets of curves in $M$. 
We define the \emph{defect group} $G^{\widehat{\bm z}}$ and its Lie algebra 
$\g^{\widehat{\bm z}}$ as
\begin{equation}\label{defect group higher}
G^{\widehat{\bm z}} \coloneqq \prod_{x \in \bm z} C^\infty(\ell \T^{n_x}_x,G) , \qquad
\g^{\widehat{\bm z}}  \coloneqq \prod_{x \in \bm z} \g \otimes_\CC \T^{n_x}_x.
\end{equation}
Since $G \subseteq \mathrm{GL}_N(\CC)$ is assumed to be a matrix Lie group,
the defect group $G^{\widehat{\bm z}}$ admits a presentation as a subgroup 
of the product  $\prod_{x \in \bm z} \mathrm{GL}_N(\mathcal T^{n_x}_x)$ of general
linear groups with entries in the Weil algebras $\mathcal T^{n_x}_x$.

We endow  $\g\otimes_\CC \mathcal T^{n_x}_x$ with the non-degenerate invariant symmetric bilinear form
\begin{equation} \label{bilinear form Tg}
\langle \cdot, \cdot \rangle : \big(\g\otimes_\CC \mathcal T^{n_x}_x\big) \times \big( \g\otimes_\CC \mathcal T^{n_x}_x) 
\longrightarrow \CC, \qquad
\langle X \otimes \varepsilon_x^p, Y \otimes \varepsilon_x^q \rangle \coloneqq k^x_{p+q}\, \langle X, Y \rangle.
\end{equation}
Non-degeneracy follows from the fact that $k^x_{n_x - 1} \neq 0$, by definition of $n_x$.
This then extends to a non-degenerate invariant symmetric bilinear form on $\g^{\widehat{\bm z}} $, 
which we denote by
\begin{equation} \label{bilinear form gD higher}
\langle\!\langle \cdot, \cdot \rangle\!\rangle_{\omega}^{} : \g^{\widehat{\bm z}}  \times \g^{\widehat{\bm z}} 
\longrightarrow \CC. 
\end{equation}
In the case when $\omega$ has only simple poles this definition reduces to \eqref{bilinear form gD}.

We have the isomorphisms
\begin{subequations} \label{defect isos higher}
\begin{align}
\label{defect isos higher a} \Omega^1(\widehat{D}, \g) &\coloneqq \prod_{x \in \bm z} \Omega^1(\Sigma, \g\otimes_\CC \mathcal T^{n_x}_x) \cong \Omega^1(\Sigma, \g^{\widehat{\bm z}} ),\\
\label{defect isos higher b} C^\infty(\widehat{D}, G) &\coloneqq \prod_{x \in \bm z} C^\infty\big( \Sigma, C^\infty(\ell \mathcal T^{n_x}_x, G) \big) \cong C^\infty(\Sigma, G^{\widehat{\bm z}} ).
\end{align}
\end{subequations}
By virtue of the isomorphism \eqref{defect isos higher a}, the pullback of smooth 
$\g$-valued 1-forms on $X$ by the morphism \eqref{D embed higher} induces a map, cf.\ \eqref{morph formal man},
\begin{equation} \label{jet map forms}
\bm j^\ast : \Omega^r(X, \g) \longrightarrow \Omega^r(\Sigma, \g^{\widehat{\bm z}}), \qquad
\eta \longmapsto \bigg(\sum_{p=0}^{n_x-1} \frac{1}{p!}  \iota_x^\ast (\partial_z^p \eta) \otimes \varepsilon_x^p\bigg)_{x \in {\bm z}},
\end{equation}
for each $r = 0, \ldots, 4$.
Likewise, the pullback of smooth $G$-valued maps on $X$ by \eqref{D embed higher} induces a map
\begin{equation} \label{jet map G functions}
\bm j^\ast  : C^\infty(X, G) \longrightarrow  C^\infty(\Sigma, G^{\widehat{\bm z}} )
, \qquad
g \longmapsto \bigg( \sum_{p=0}^{n_x-1} \frac{1}{p!}  \iota_x^\ast (\partial_z^p g) \otimes \varepsilon_x^p \bigg)_{x \in \bm z},
\end{equation}
where the presentation $G^{\widehat{\bm z}} \subseteq \prod_{x \in \bm z} \mathrm{GL}_N(\mathcal T^{n_x}_x)$ 
as a matrix Lie group is understood. Using the Leibniz rule, one easily proves that 
$\bm j^\ast$ is a group homomorphism, i.e.\
\begin{equation}\label{jast is group homomorphism}
{\bm j}^\ast(g'\,g) = ({\bm j}^\ast g')\, ({\bm j}^\ast g),
\end{equation}
for all $g,g'\in C^\infty(X, G)  $.

The following result extends Lemma \ref{lem: CP formula} to the case of higher order poles.
\begin{lemma} \label{lem: int omega d eta higher}
For any $\zeta= \sum_{p=0}^{n-1} \zeta_p \otimes \varepsilon^p \in \Omega^2(X) \otimes_\CC \T^n$, we have
\begin{equation*}
\int_X ( \omega \wedge d \zeta)_{\rm reg}^{} = 2 \pi \ii \sum_{x \in \bm z} \sum_{p=0}^{n_x - 1} k^x_p \int_{\Sigma_x} \iota_x^\ast \zeta_p.
\end{equation*}
\end{lemma}
\begin{proof}
Since $d\zeta = \sum_{p=0}^{n-1} d\zeta_p \otimes \varepsilon^p$, using the definition \eqref{reg def} we find
\begin{equation*}
\int_X ( \omega \wedge d \zeta)_{\rm reg}^{} =
\sum_{x \in \bm z} \sum_{p=0}^{n_x - 1} \int_X \frac{k^x_p\, dz}{z - x} \wedge d \zeta_p
= 2 \pi \ii \sum_{x \in \bm z} \sum_{p=0}^{n_x - 1} k^x_p \int_{\Sigma_x} \iota_x^\ast \zeta_p,
\end{equation*}
where in the second equality we used Lemma \ref{lem: CP formula}.
\end{proof}

We may now rewrite the second term on the right hand side of \eqref{reg S4d gauge transf} as follows.
\begin{proposition} \label{prop: kinetic higher pole}
For any $g \in C^\infty(X, G)$ and $A \in \Omega^1(X, \g)$, we have
\begin{equation*}
\int_X \big( \omega \wedge j_X^\ast d \langle g^{-1} d g , A \rangle \big)_{\rm reg}^{}
= 2 \pi \ii \int_\Sigma  \big\langle{\mkern-5mu}\big\langle (\bm j^\ast g)^{-1} d_\Sigma (\bm j^\ast g), \bm j^\ast A \big\rangle{\mkern-5mu}\big\rangle_{\omega}.
\end{equation*}
\end{proposition}
\begin{proof}
It follows from Lemma \ref{lem: int omega d eta higher} that
\begin{equation*}
\int_X \big( \omega \wedge j_X^\ast d \langle g^{-1} d g , A \rangle \big)_{\rm reg} = 2 \pi \ii \sum_{x \in \bm z} \sum_{p=0}^{n_x - 1} k^x_p \int_{\Sigma_x} \iota_x^\ast \bigg( \frac{1}{p!} \partial_z^p \big\langle g^{-1} d g , A \big\rangle \bigg).
\end{equation*}
Applying the Leibniz rule to the right hand side, we find
\begin{align*}
&2 \pi \ii \sum_{x \in \bm z} \sum_{p=0}^{n_x - 1} \sum_{r=0}^p \int_{\Sigma_x} k^x_p\, \bigg\langle \frac{1}{r!} \iota_x^\ast \big( \partial_z^r (g^{-1} d g) \big), \frac{1}{(p-r)!} \iota_x^\ast \big( \partial_z^{p-r} A \big) \bigg\rangle\\
&\quad = 2 \pi \ii \sum_{x \in \bm z} \sum_{p=0}^{n_x - 1} \sum_{r=0}^p \int_{\Sigma_x} \bigg\langle \frac{1}{r!} \iota_x^\ast \big( \partial_z^r (g^{-1} d g) \big) \otimes \varepsilon_x^r, \frac{1}{(p-r)!} \iota_x^\ast \big( \partial_z^{p-r} A \big) \otimes \varepsilon_x^{p-r} \bigg\rangle\\
&\quad = 2 \pi \ii \int_\Sigma \big\langle{\mkern-5mu}\big\langle {\bm j}^\ast (g^{-1}d g), {\bm j}^\ast A \big\rangle{\mkern-5mu}\big\rangle_{\omega}^{}
= 2 \pi \ii \int_\Sigma \big\langle{\mkern-5mu}\big\langle ({\bm j}^\ast g)^{-1} d_\Sigma ({\bm j}^\ast g), {\bm j}^\ast A \big\rangle{\mkern-5mu}\big\rangle_{\omega}^{},
\end{align*}
where in the first equality we used the definition of the bilinear form \eqref{bilinear form Tg} on $\g\otimes_\CC \T^{n_x}_x$
and in the second equality the definition of the bilinear form \eqref{bilinear form gD higher}.
The last equality follows from the identity
$ {\bm j}^\ast(g^{-1}d g) = ({\bm j}^\ast g)^{-1} d_\Sigma ({\bm j}^\ast g)$,
which is proved similarly to \eqref{jast is group homomorphism} by a simple Leibniz rule argument.
\end{proof}

We now turn to the third term on the right hand side of \eqref{reg S4d gauge transf},
which requires some preparation. We denote by $\widehat{G} \coloneqq C^\infty(\ell \mathcal{T}^n,G)$ 
the mapping space from the locus of the Weil algebra $\mathcal{T}^n=\CC[\varepsilon]/(\varepsilon^n)$ 
to $G$, where we recall that $n=\max\; \{n_x\}_{x\in{\bm z}}$ is the maximal order of all poles. 
Note that $\widehat{G}$ is the Lie group of $(n-1)$-jets of curves in $G$
and that its Lie algebra is $\widehat{\g} = \g\otimes_\CC \mathcal{T}^n$.
Analogously to \eqref{jet prolongation}, we introduce the map 
\begin{equation*}
j_X^\ast : C^\infty(X,G) \longrightarrow C^\infty(X,\widehat{G}),\qquad 
g\longmapsto \sum_{p=0}^{n-1}\frac{1}{p!} \partial_z^p g \otimes \varepsilon^p,
\end{equation*}
which describes the holomorphic part of the $(n - 1)$-jet prolongation 
of smooth $G$-valued maps on $X$.  Using the Leibniz rule, one easily proves that 
$j_X^\ast$ is a group homomorphism, i.e.\
$ j_X^\ast(g'\,g) = (j_X^\ast g')\, (j_X^\ast g)$, for all $g,g'\in C^\infty(X, G) $.
Using again the Leibniz rule, one further shows that the
$\mathcal{T}^n$-valued form $ j_X^\ast (g^\ast \chi_G)\in \Omega^3(X)\otimes_\CC \mathcal{T}^n $ 
in \eqref{reg S4d gauge transf} can be expressed as
\begin{equation}\label{widehat chi}
 j_X^\ast (g^\ast \chi_G)  = (j_X^\ast g)^\ast\overline{\chi}_{\widehat{G}},
\end{equation}
where $\overline{\chi}_{\widehat{G}} \coloneqq\frac{1}{6} \langle \theta_{\widehat{G}}, [\theta_{\widehat{G}},\theta_{\widehat{G}}]\rangle\in 
\Omega^3(\widehat{G})\otimes_\CC\mathcal{T}^n$ is the $\mathcal{T}^n$-valued Cartan $3$-form
defined by the $\mathcal{T}^n$-bilinear extension $ \langle \cdot,\cdot\rangle : \widehat{\g}\times\widehat{\g} \to\mathcal{T}^n$
of the bilinear form $\langle \cdot,\cdot\rangle$ on $\g$.

The generalisation of Proposition \ref{prop: WZ depends on D} to the higher pole case
requires a suitable modification of the Lemmas \ref{lem: smooth vs continuous}, \ref{lem: contracting R2} 
and \ref{lem: singleton} to maps with values in the jet group $\widehat{G} = C^\infty(\ell \mathcal{T}^n,G)$.
Let us start by highlighting the commutative diagram
\begin{equation} \label{bf j vs. j_X}
\begin{tikzcd}
C^\infty(X,G) \arrow[r, "j_X^\ast"] \arrow[d, "{\bm j}^\ast"'] & C^\infty(X,\widehat{G}) \arrow[d, "{\bm \iota}^\ast"]\\
C^\infty(\Sigma,G^{\widehat{\bm z}}) & C^\infty(\Sigma, \widehat{G}^{\bm z} ) \arrow[l, "\mathrm{trunc}"]
\end{tikzcd}
\end{equation}
where the map $\mathrm{trunc}$ is given by post-composition with the map
\begin{equation}\label{jet truncation map}
\widehat{G}^{\bm z}  = \prod_{x\in{\bm z}} C^\infty(\ell \mathcal{T}^n,G) \longrightarrow
\prod_{x\in{\bm z}} C^\infty(\ell \mathcal{T}^{n_x}_x,G) = G^{\widehat{\bm z}}
\end{equation}
that truncates the orders of jets. (Recall that by definition $n_x\leq n$, for all $x\in{\bm z}$.)
We generalise the concepts of relative maps and relative homotopies from \S\ref{sec: gauge inv action} by
\begin{subequations}\label{Cinfty hat D}
\begin{equation} \label{Cinfty hat D a}
C^\infty_{\widehat{D}}(X,\widehat{G}) \coloneqq\big\{g\in C^\infty(X,\widehat{G})\,:\, \mathrm{trunc} \,{\bm \iota}^\ast g = e\big\}
\end{equation}
and
\begin{equation} \label{Cinfty hat D b}
C^\infty_{\widehat{D}\times I}(X\times I,\widehat{G}) \coloneqq\big\{H\in C^\infty(X\times I,\widehat{G})\,:\, \mathrm{trunc} \,({\bm \iota}\times \id)^\ast H = e\big\},
\end{equation}
\end{subequations}
where $I=[0,1]$ is the unit interval. We denote by
$C^\infty_{\widehat{D}}(X,\widehat{G}) \big/\!\!\sim_{\widehat{D}}^\infty$ 
the corresponding set of homotopy classes.
\begin{lemma}\label{lem: hat G homotopy}
$C^\infty_{\widehat{D}}(X,\widehat{G})\big /\!\!\sim_{\widehat{D}}^\infty$ is a singleton.
\end{lemma}
\begin{proof}
We recall from \cite{Vizmann} that there exists, for each $k\geq 1$,
a diffeomorphism $C^\infty(\ell \mathcal{T}^k,G) \cong G\times \g^{k-1}$ between the $(k-1)$-jet group 
and a Cartesian product of $G$ with $k-1$ copies of the Lie algebra $\g$. 
Under these diffeomorphisms, the maps $\widehat{G} =C^\infty(\ell \mathcal{T}^{n},G)\to C^\infty(\ell \mathcal{T}^{n_x}_x,G)$
truncating the jet orders are given by projection maps $G \times\g^{n-1}\to G\times\g^{n_x-1}$
onto the first $n_x$ factors. From the universal property of products and the definition \eqref{Cinfty hat D}, 
one obtains that
\begin{equation}\label{product homotopy classes}
C^\infty_{\widehat{D}}(X,\widehat{G}) \big/\!\!\sim_{\widehat{D}}^\infty ~\,\cong~ C^\infty_D(X,G)\big/\!\!\sim_D^\infty 
\times \prod_{i=1}^{n-1}  C^\infty_{D_i}(X,\g)\big/\!\!\sim_{D_i}^\infty
\end{equation}
is a product of sets of relative homotopy classes of maps as in \S\ref{sec: gauge inv action},
where $D = \bigsqcup_{x\in{\bm z}} \Sigma_x$ is the non-thickened defect and
\begin{equation*}
D_i \coloneqq \bigsqcup_{x\in {\bm z}\,:\, n_x-1\geq i} \Sigma_x
\end{equation*} 
is the disjoint union of the non-thickened
connected components of the defect $\widehat{D}$ 
that support $i$-jet data, for $i=1,\ldots,n-1$. By the same arguments as in the 
proofs of Lemmas \ref{lem: smooth vs continuous}, \ref{lem: contracting R2} 
and \ref{lem: singleton}, one shows that each factor on the right hand side of 
\eqref{product homotopy classes} is a singleton. Hence, their product is a singleton too.
\end{proof}

The following result is the generalisation of Proposition \ref{prop: WZ depends on D} to the case of higher order poles.
\begin{proposition} \label{prop: WZ depends on D higher}
The integral $\int_X \big( \omega \wedge \widetilde{g}^\ast \overline{\chi}_{\widehat{G}} \big)_{\rm reg}^{}$ depends on
$\widetilde{g}\in C^\infty(X,\widehat{G})$ only through $\mathrm{trunc}\, {\bm \iota}^\ast\widetilde{g} 
\in C^\infty(\Sigma, G^{\widehat{\bm z}})$. In particular,
$\int_X \big( \omega \wedge j_X^\ast (g^\ast \chi_G) \big)_{\rm reg}^{}$ depends on
$g\in C^\infty(X,G)$ only through $\bm j^\ast g \in C^\infty(\Sigma, G^{\widehat{\bm z}})$.
\end{proposition}
\begin{proof}
This is very similar to the proof of Proposition \ref{prop: WZ depends on D}. We refer to the latter for 
certain details, highlighting only the parts of the proof which are different in the present higher order pole setting.

Let $\widetilde{g}, \widetilde{h} \in C^\infty(X, \widehat{G})$ be such that 
$\mathrm{trunc}\, {\bm \iota}^\ast\widetilde{g}  = \mathrm{trunc}\, {\bm \iota}^\ast\widetilde{h} $.
From the Polyakov-Wiegmann identity and 
an argument as in the proof of Proposition \ref{prop: kinetic higher pole}, we obtain 
\begin{equation*}
\int_X \big( \omega \wedge \widetilde{g}^\ast \overline{\chi}_{\widehat{G}}  \big)_{\rm reg}^{} -
\int_X \big( \omega \wedge \widetilde{h}^\ast \overline{\chi}_{\widehat{G}}  \big)_{\rm reg}^{}
=\int_X \big( \omega \wedge (\widetilde{g}\widetilde{h}^{-1})^\ast \overline{\chi}_{\widehat{G}}  \big)_{\rm reg}^{}.
\end{equation*}
It remains to prove that the right hand side of this equation vanishes, provided that
$\mathrm{trunc}\, {\bm \iota}^\ast\widetilde{g}  = \mathrm{trunc}\, {\bm \iota}^\ast\widetilde{h} $,
which by \eqref{Cinfty hat D a} is equivalent to $\widetilde{g}\widetilde{h}^{-1}\in C^\infty_{\widehat{D}}(X,\widehat{G})$.
It follows from Lemma \ref{lem: hat G homotopy} that there exists a homotopy
$H\in  C^\infty_{\widehat{D}\times I}(X\times I,\widehat{G}) $
between $\widetilde{g}\widetilde{h}^{-1}$ and $e \in C^\infty_{\widehat{D}}(X,\widehat{G})$.
We deduce that
\begin{equation*}
(\widetilde{g}\widetilde{h}^{-1})^\ast \overline{\chi}_{\widehat{G}}
= d \bigg( \int_I H^\ast\overline{\chi}_{\widehat{G}} \bigg)
\end{equation*}
by the same line of arguments as in \eqref{WZ g h inv exact}. 
It then follows by using Lemma \ref{lem: int omega d eta higher} that
\begin{equation*}
\int_X \big( \omega \wedge (\widetilde{g}\widetilde{h}^{-1})^\ast \overline{\chi}_{\widehat{G}} \big)_{\rm reg}^{} = 
2 \pi \ii \sum_{x \in \bm z} \sum_{p=0}^{n_x-1} k^x_p \int_{\Sigma_x \times I} (\iota_x \times \id)^\ast( H^\ast\overline{\chi}_{\widehat{G}} )_p^{}
=0,
\end{equation*}
where the last equality follows from
$\mathrm{trunc} \,({\bm \iota}\times\id)^\ast H = e\in C^\infty(\Sigma,G^{\widehat{\bm z}})$
by definition of  $H\in  C^\infty_{\widehat{D}\times I}(X\times I,\widehat{G}) $, cf.\ \eqref{Cinfty hat D b}.

The special case in the statement of this proposition is a consequence of \eqref{widehat chi} and \eqref{bf j vs. j_X}.
\end{proof}

We can now prove the generalisation of Proposition \ref{prop: int WZ compute} to the present setting.
\begin{proposition} \label{prop: int WZ compute higher}
For any $g \in C^\infty(X, G)$, we have
\begin{equation*}
\int_X \big( \omega \wedge j_X^\ast (g^\ast \chi_G) \big)_{\rm reg} = 
- 2 \pi \ii \int_{\Sigma \times I} \widehat{g}^\ast \chi_{G^{\widehat{\bm z}}},
\end{equation*}
where $\chi_{G^{\widehat{\bm z}}}\in\Omega^3(G^{\widehat{\bm z}})$ 
is the Cartan $3$-form on $G^{\widehat{\bm z}}$
and $\widehat{g} \in C^\infty(\Sigma \times I, G^{\widehat{\bm z}})$ 
is any lazy homotopy between $\bm j^\ast g \in C^\infty(\Sigma, G^{\widehat{\bm z}})$ 
and the constant map $e \in C^\infty(\Sigma, G^{\widehat{\bm z}})$.
\end{proposition}
\begin{proof}
The argument is an adaptation of the proof of Proposition \ref{prop: int WZ compute}
to the case of higher order poles. We thus refer to the latter for certain details and highlight only the
features pertaining to the present case.

Since $G^{\widehat{\bm z}}$ is connected and $\Sigma = \RR^2$ is contractible, 
there exists a lazy homotopy $\widehat{g} \in C^\infty(\Sigma \times I, G^{\widehat{\bm z}})$
between ${\bm j}^{\ast} g$ and $e$. Using the fact that the jet order truncation map 
$\widehat{G}^{\bm z}\to G^{\widehat{\bm z}}$ in \eqref{jet truncation map} is a trivial fibre bundle \cite{Vizmann}, 
we can lift $\widehat{g}$ to a lazy homotopy $\overline{g} \in C^\infty(\Sigma \times I, \widehat{G}^{\bm z})$
between a lift of ${\bm j}^{\ast} g$  and the identity element $e\in C^\infty(\Sigma \times I, \widehat{G}^{\bm z})$.
By construction, $\mathrm{trunc}\,\overline{g} = \widehat{g}$.

As in the proof of Proposition \ref{prop: int WZ compute},
let $\varrho : \Delta \to I$ denote the radial coordinate on the unit disc $\Delta$
and let $\Delta_x \subset C$ be disjoint discs around each pole $x \in \bm z$. 
We define $\widetilde g \in C^\infty(X, \widehat{G})$ as the extension from 
$\bigsqcup_{x \in \bm z} \Sigma \times \Delta_x$ to $X$ by the identity 
of the image of $(\id_\Sigma \times \varrho)^\ast \overline{g} \in C^\infty(\Sigma \times \Delta, \widehat{G}^{\bm z})$ 
under the isomorphism
\begin{equation*}
C^\infty\bigg( \bigsqcup_{x \in \bm z} \Sigma \times \Delta_x, \widehat{G} \bigg)
\cong C^\infty(\Sigma \times \Delta, \widehat{G}^{\bm z}).
\end{equation*}
By construction, we have that $ \mathrm{trunc}\, {\bm \iota}^\ast\widetilde{g} = {\bm j}^{\ast} g = 
\mathrm{trunc}\, {\bm \iota}^\ast j_X^\ast g$, hence Proposition \ref{prop: WZ depends on D higher}
implies that
\begin{equation*}
\int_X \big( \omega \wedge j_X^\ast (g^\ast \chi_G) \big)_{\rm reg}^{} = 
\int_X \big( \omega \wedge (j_X^\ast g)^\ast \overline{\chi}_{\widehat{G}} \big)_{\rm reg}^{} 
= \int_X \big( \omega \wedge \widetilde{g}^\ast \overline{\chi}_{\widehat{G}} \big)_{\rm reg}^{}.
\end{equation*}
The integral on the right hand side can be computed by following the same steps 
as in the end of the proof of Proposition \ref{prop: int WZ compute}. Explicitly, we have
\begin{align*}
\int_X \big( \omega \wedge \widetilde{g}^\ast \overline{\chi}_{\widehat{G}} \big)_{\rm reg}^{} 
&= \sum_{x \in \bm z} \int_{\Sigma \times \Delta_x} \big( \omega \wedge \widetilde{g}^\ast \overline{\chi}_{\widehat{G}} \big)_{\rm reg}^{}
= \sum_{x \in \bm z} \int_{\Sigma \times \Delta_x} \Big( \omega \wedge d \Big( - \int_{\gamma_z} \widetilde{g}^\ast \overline{\chi}_{\widehat{G}}\Big) \Big)_{\rm reg}^{}\\
&= 
- 2 \pi \ii \sum_{x \in \bm z} \sum_{p=0}^{n_x-1} k^x_p \int_{\Sigma_x} \int_{\gamma_x} (\widetilde{g}^\ast \overline{\chi}_{\widehat{G}})_p^{} 
= - 2 \pi \ii \int_{\Sigma \times I} \widehat{g}^\ast \chi_{G^{\widehat{\bm z}}}.
\end{align*}
In the third equality we used Lemma \ref{lem: int omega d eta higher}
and in the last step we used the definition of the bilinear form $\langle\!\langle \cdot,\cdot\rangle\!\rangle_{\omega}^{}$ 
on $\g^{\widehat{\bm z}}$ from \eqref{bilinear form Tg} and \eqref{bilinear form gD higher}.
\end{proof}

\section{Boundary conditions on surface defects} \label{sec: boundary conditions}
The results in this section are stated and proved for poles of arbitrary orders $n_x\geq 1$.
We use our notational conventions from the higher order pole case in \S\ref{sec: higher poles}. 
The definitions and results in \S\ref{sec: higher poles} reduce to the ones 
in \S\ref{sec: simple pole case} in the case when all poles are simple,
i.e.\ $n_x =1$, for all $x\in {\bm z}$, and consequently $n=1$.

\subsection{Bulk fields with boundary conditions} \label{sec: bulk with bc}
We introduce a groupoid of bulk fields with boundary conditions 
at the (thickened) surface defect $\widehat{D}$. Imposing these boundary conditions will have the 
effect of making the action \eqref{reg 4dCS action} gauge invariant.

To define the relevant groupoid, let us first observe that the action
\eqref{reg 4dCS action} is invariant under translations by $\g$-valued $(0,1,0)$-forms, i.e.\
\begin{equation*}
S_{\omega}(A+\lambda) = S_{\omega}(A),
\end{equation*}
for all $A\in\Omega^1(X,\g)$ and $\lambda\in \Omega^{0,1,0}(X,\g)$,
which is due to the fact that $\omega \in \Omega^{0,1,0}(X)$. Hence, the action descends 
to the quotient
\begin{equation}\label{Omega quotient iso}
\overline{\Omega}^1(X,\g) \coloneqq \frac{\Omega^1(X,\g)}{\Omega^{0,1,0}(X,\g)} \cong \Omega^{1,0,0}(X,\g) \oplus \Omega^{0,0,1}(X,\g),
\end{equation}
where the last isomorphism is due to the direct sum decomposition \eqref{triple graded forms} of forms on $X$.
The gauge transformations in \eqref{gauge transf} also descend to the quotient,
because for every $g\in C^\infty(X,G)$ and $\lambda\in \Omega^{0,1,0}(X,\g)$ we have
\begin{equation*}
\null^g(A+\lambda) = -dg g^{-1} + gAg^{-1} + g \lambda g^{-1} =\null^g A + g\lambda g^{-1}
\end{equation*}
and $g \lambda g^{-1} \in \Omega^{0,1,0}(X,\g)$.
Abusing notation slightly, we will denote also by $\null^g A$ the action of a gauge transformation
$g\in C^\infty(X,G)$ on a 1-form $A \in \overline{\Omega}^1(X,\g)$ under the isomorphism in \eqref{Omega quotient iso}, which explicitly reads
\begin{equation*}
\null^g A =- \overline{d}g g^{-1}  + g A g^{-1}, 
\end{equation*}
where $\overline{d} \coloneqq d_\Sigma + \bar \partial$.

We define the groupoid of \emph{bulk fields} on $X$ by
\begin{equation*}
\bm{B}G_{\rm con}(X) \coloneqq \left\{
\begin{array}{ll}
\textup{Obj}: & A \in \overline{\Omega}^1(X, \g),\\
\textup{Mor}: & g: A \to \null^g A, \;\; \textup{with} \; g \in C^\infty(X, G),
\end{array}
\right.
\end{equation*}
and the groupoid of \emph{defect fields} on $\widehat{D}$ by
\begin{equation*}
\bm{B}G^{\widehat{\bm z}}_{\rm con}(\Sigma) \coloneqq \left\{
\begin{array}{ll}
\textup{Obj}: & a \in \Omega^1(\Sigma, \g^{\widehat{\bm z}}),\\
\textup{Mor}: & k: a \to \null^k a, \;\; \textup{with} \; k \in C^\infty(\Sigma, G^{\widehat{\bm z}}),
\end{array}
\right.
\end{equation*}
where $G^{\widehat{\bm z}}$ is the defect group and $\g^{\widehat{\bm z}}$ its Lie algebra,
cf.\ \eqref{defect group higher}. We would like to emphasise that there is no need to introduce different bundles
in these groupoids, because every principal $G$-bundle on $X$ and every principal 
$G^{\widehat{\bm z}}$-bundle on $\Sigma$ is trivialisable. The latter follows from 
$\Sigma=\RR^2$ being homotopic to a point, while the former follows from the existence 
of a deformation retract from $X$ to a bouquet of circles $\bigvee^{\vert{\bm \zeta}\vert-1} S^1$
and the short calculation
\begin{equation*}
\pi_0\big(\Map_{\{a\}}(X,{\bm B}G)\big) \cong \pi_0\big(\Map_{\{a\}}(S^1,{\bm B}G)\big)^{\vert{\bm \zeta}\vert-1} \cong
\pi_0(G)^{\vert{\bm \zeta}\vert-1} \cong \{\ast\},
\end{equation*}
where $a\in X$ is any choice of base point and ${\bm B} G$ denotes the classifying space of principal $G$-bundles.
The last isomorphism follows since $G$ is connected.

Using \eqref{jet map forms} for $r=1$ and \eqref{jet map G functions}, 
we introduce the functor 
\begin{equation*}
\bm j^\ast : \bm{B}G_{\rm con}(X) \longrightarrow \bm{B}G^{\widehat{\bm z}}_{\rm con}(\Sigma)
\end{equation*}
that sends an object $A$ to $\bm j^\ast A$ (note that this is well-defined 
on the quotients in \eqref{Omega quotient iso}) and a morphism 
$g : A \to \null^g A$ to $\bm j^\ast g : \bm j^\ast A \to \bm j^\ast(\null^g A) 
= \null^{\bm j^\ast g} (\bm j^\ast A)$.

\medskip

In order to impose boundary conditions for the field $A \in \overline{\Omega}^1(X, \g)$ 
on the surface defect $\widehat{D}$, we introduce a subgroupoid of $\bm BG^{\widehat{\bm z}}_{\rm con}(\Sigma)$ as follows.
Fix a Lie subalgebra $\k \subset \g^{\widehat{\bm z}}$, which
is isotropic with respect to the bilinear form $\langle\!\langle \cdot, \cdot \rangle\!\rangle_{\omega}^{}$ 
in \eqref{bilinear form gD higher}, and let $K \subset G^{\widehat{\bm z}}$ 
denote the corresponding connected Lie subgroup. We define 
\begin{equation*}
\bm{B}K_{\rm con}(\Sigma) \coloneqq \left\{
\begin{array}{ll}
\textup{Obj}: & a \in \Omega^1(\Sigma, \k) \subset \Omega^1(\Sigma, \g^{\widehat{\bm z}}),\\
\textup{Mor}: & k: a \to \null^k a, \;\; \textup{with} \; k \in C^\infty(\Sigma, K) \subset C^\infty(\Sigma, G^{\widehat{\bm z}}),
\end{array}
\right.
\end{equation*}
and observe that, by definition, there is an inclusion functor 
\begin{equation*}
\bm{B}K_{\rm con}(\Sigma) \longhookrightarrow \bm{B}G^{\widehat{\bm z}}_{\rm con}(\Sigma).
\end{equation*}
Given such a choice of an isotropic Lie subalgebra $\k \subset \g^{\widehat{\bm z}}$, we define
the groupoid of \emph{bulk fields with boundary conditions} by
\begin{equation} \label{Fbc def}
\mathfrak{F}_{\rm bc}(X) \coloneqq \left\{
\begin{array}{ll}
\textup{Obj}: & A \in \overline{\Omega}^1(X, \g), \;\; \textup{s.t.} \; \bm j^\ast A \in \Omega^1(\Sigma, \k),\\
\textup{Mor}: & g: A \to \null^g A, \;\; \textup{with} \; g \in C^\infty(X, G)\; \textup{s.t.} \; \bm j^\ast g \in C^\infty(\Sigma, K).
\end{array}
\right.
\end{equation}
Given any morphism $g: A \to \null^g A$ in $\mathfrak{F}_{\rm bc}(X)$, we have 
$(\bm j^\ast g)^{-1} d_\Sigma (\bm j^\ast g) \in \Omega^1(\Sigma, \k)$ and $\bm j^\ast A \in \Omega^1(\Sigma, \k)$. 
Hence, the second term on the right hand side of \eqref{reg S4d gauge transf} vanishes 
on account of Proposition \ref{prop: kinetic higher pole} and the isotropy of $\k \subset \g^{\widehat{\bm z}}$. 
The proposition below shows that the last term on the right hand side of 
\eqref{reg S4d gauge transf} also vanishes.

\begin{proposition}
$\int_X \big( \omega \wedge j^\ast_X (g^\ast \chi_G) \big)_{\rm reg} = 0$,
for every morphism $g: A \to \null^g A$ in $\mathfrak{F}_{\rm bc}(X)$.
\end{proposition}
\begin{proof}
By Proposition \ref{prop: int WZ compute higher}, we have
\begin{equation*}
\int_X \big( \omega \wedge j_X^\ast (g^\ast \chi_G) \big)_{\rm reg}^{} = - 2 \pi \ii \int_{\Sigma \times I} \widehat{g}^\ast \chi_{G^{\widehat{\bm z}}},
\end{equation*}
where $\widehat{g} \in C^\infty(\Sigma \times I, G^{\widehat{\bm z}})$ is any lazy homotopy between 
$\bm j^\ast g \in C^\infty(\Sigma, K)$ and $e \in C^\infty(\Sigma, K)$. 
Since $K$ is connected, we can choose a lazy homotopy 
$\widehat{g}$ with values in the Lie subgroup $K \subset G^{\widehat{\bm z}}$, 
i.e.\ $\widehat{g} \in C^\infty(\Sigma \times I, K)$.
It then follows that $\widehat{g}^{-1} d_{\Sigma \times I} \widehat{g} \in \Omega^1(\Sigma \times I, \k)$ 
and therefore $\widehat{g}^\ast \chi_{G^{\widehat{\bm z}}} = 0$ since $\k \subset \g^{\widehat{\bm z}}$ 
is an isotropic Lie subalgebra.
\end{proof}

Summing up, we obtain
\begin{theorem}\label{cor gauge invariance}
The regularised $4$-dimensional Chern-Simons action $S_\omega$ given in \eqref{reg 4dCS action} 
defines a gauge invariant action on the groupoid 
$\mathfrak{F}_{\rm bc}(X)$. 
\end{theorem}

To conclude, we would like to note that the groupoid $\mathfrak{F}_{\rm bc}(X)$ 
in \eqref{Fbc def} is a model for the pullback
\begin{equation} \label{strict pullback}
\begin{tikzcd}
\mathfrak{F}_{\rm bc}(X) \arrow[r, dashed] \arrow[d, dashed] & \bm{B}G_{\rm con}(X) \arrow[d, "{\bm j^\ast}"]\\
\bm{B}K_{\rm con}(\Sigma) \arrow[r, hookrightarrow] & \bm{B}G^{\widehat{\bm z}}_{\rm con}(\Sigma)
\end{tikzcd}
\end{equation}
in the category of groupoids. This fact motivates our construction in the next subsection.

\subsection{Bulk fields with edge modes}
The category of groupoids is a category with weak equivalences, 
where the latter are given by equivalences of groupoids, 
i.e.\ fully faithful and essentially surjective functors. In general, pullbacks fail to 
preserve weak equivalences. This means that if we were to replace the pullback diagram in 
\eqref{strict pullback} by a weakly equivalent one, 
in general its pullback will not be weakly equivalent 
to $\mathfrak{F}_{\rm bc}(X)$. To solve this issue 
one considers homotopy pullbacks, instead of ordinary categorical pullbacks, 
which do preserve weak equivalences.
We refer to \cite{Hovey,Riehl} for an introduction to the frameworks
of model and homotopical category theory that underlies the
study of homotopy pullbacks.

Motivated by the above discussion, we define the \emph{field groupoid} $\mathfrak{F}(X)$ 
as the homotopy pullback 
\begin{equation} \label{homotopy pullback}
\begin{tikzcd}
\mathfrak{F}(X) \arrow[r, dashed] \arrow[d, dashed] \arrow[dr, phantom, "\scalebox{0.8}{$h$}" , pos=0.04] & \bm{B}G_{\rm con}(X) \arrow[d, "{\bm j^\ast}"]\\
\bm{B}K_{\rm con}(\Sigma) \arrow[r, hookrightarrow] & \bm{B}G^{\widehat{\bm z}}_{\rm con}(\Sigma)
\end{tikzcd}
\end{equation}
in the model category of groupoids. Computing this homotopy pullback
by a standard construction (see e.g.\ \cite[Appendix A]{HomEdgeModes} for a review),
we obtain
\begin{equation} \label{FM def}
\mathfrak{F}(X) \coloneqq \left\{
\begin{array}{ll}
\textup{Obj}: & (A, h) \in \overline{\Omega}^1(X, \g) \times C^\infty(\Sigma, G^{\widehat{\bm z}}), \;\; \textup{s.t.} \; \null^{h^{-1}} (\bm j^\ast A) \in \Omega^1(\Sigma,\k),\\
\textup{Mor}: & (g,k) :(A, h) \to ( \null^g A, ({\bm j}^\ast g) h k^{-1}), \;\;   \\
&\textup{with } \;g \in C^\infty(X, G)\; \textup{ and }\; k\in C^\infty(\Sigma, K).
\end{array}
\right.
\end{equation}
This is to be compared with the (strict) pullback $\mathfrak{F}_{\rm bc}(X)$ in \eqref{Fbc def}.
\begin{theorem} \label{thm: equiv Fbc to F}
The functor
\begin{align*}
\Phi: \mathfrak{F}_{\rm bc}(X) \longrightarrow \mathfrak{F}(X)
\end{align*}
that sends an object $A$ to $(A,e)$ and a morphism $g: A \to \null^g A$ 
to $(g,{\bm j}^\ast g): (A,e) \to (\null^g A,e)$ is an equivalence of groupoids.
\end{theorem}
\begin{proof}
$\Phi$ is obviously faithful. To show that it is also full, 
consider objects $A, A' \in \mathfrak{F}_{\rm bc}(X)$ and 
let $(g,k): \Phi(A) =(A,e) \to \Phi(A^\prime)=(A',e)$ be a morphism in $\mathfrak{F}(X)$. 
By definition, $A^\prime = \null^g A$ and $(\bm j^\ast g) k^{-1} =e$, i.e.\ $\bm j^\ast g = k \in C^\infty(\Sigma,K)$. 
This shows that $g: A \to A^\prime$ is a morphism 
in $\mathfrak{F}_{\rm bc}(X)$ and, indeed, $\Phi(g) = (g,k)$. 

To conclude the proof, we have to show that $\Phi$ is essentially surjective. 
Let $(A,h) \in \mathfrak{F}(X)$. 
Recall that the jet order truncation map 
$\widehat{G}^{\bm z}\to G^{\widehat{\bm z}}$ in \eqref{jet truncation map} is a trivial fibre bundle \cite{Vizmann} 
and consider a lift $\widehat h \in C^\infty(\Sigma,\widehat G^{\bm z})$ of 
$h \in C^\infty(\Sigma, K) \subset C^\infty(\Sigma, G^{\widehat{\bm z}})$. 
By the construction in the proof of Proposition \ref{prop: int WZ compute} 
(just consider the Lie group $\widehat G$ instead of $G$, noting that $\widehat{G}$ is connected since $G$ is), we obtain 
an extension $\widetilde{h} \in C^\infty(X,\widehat G)$ of $\widehat h$, i.e.\ such that $\bm \iota^\ast \widetilde{h} = \widehat{h}$,
with the following properties:
\begin{itemize}
  \item[(a)] the restriction of $\widetilde{h}$ to 
$\bigsqcup_{x \in \bm z} \Sigma \times \Delta_x \subset X$ 
is constant along $C$, where each $\Delta_x \subset C$ is a 
sufficiently small open disc centred at $x \in \bm z$, and
  \item[(b)] $\widetilde{h}$ takes the constant value $e \in \widehat G$ on an open
neighbourhood of $X \setminus \bigsqcup_{x \in \bm z} \Sigma \times \Delta^\prime_x$, 
where each $\Delta^\prime_x \supsetneq \Delta_x$ is a strictly larger open disc centred at $x \in \bm z$.
\end{itemize}
Using the diffeomorphism $\widehat G \cong G \times \g^{n-1}$ from \cite{Vizmann}, 
$\widetilde{h} \in C^\infty(X,\widehat G)$ can also be regarded as a tuple of maps on $X$, 
where $\widetilde h_0 \in C^\infty(X,G)$ is $G$-valued and 
$\widetilde h_i \in C^\infty(X,\g)$ is $\g$-valued, for $i=1,\ldots,n-1$. 
Below we use these data to construct $g \in C^\infty(X,G)$ 
such that $\bm{\iota}^\ast j_X^\ast g = \widehat h$. 
For each $x \in \bm z$, consider the local coordinate $z-x$ 
on $\Delta^\prime_x$ centred at $x$ 
and define $g$ on $\Sigma \times \Delta^\prime_x$ by 
\begin{equation*}
g \coloneqq \exp\bigg( \sum_{i=1}^{n-1}\frac{ (z-x)^i}{i!} \,\xi_i \bigg) \widetilde h_0,
\end{equation*}
where each $\xi_i \in C^\infty(\Sigma\times \Delta^\prime_x,\g)$, 
for $i=1,\ldots,n-1$, is a linear combination of the $\widetilde h_i$'s
and of their Lie brackets. 
Arguing by induction on $i$, the explicit expression of the $\xi_i$'s 
is obtained by imposing the condition 
$\bm{\iota}^\ast j_X^\ast g = \widehat{h}$. 
(Explicitly, using (a) one finds $\xi_1 \coloneqq \widetilde h_1$, $\xi_2 \coloneqq \widetilde h_2$, 
$\xi_3 \coloneqq \widetilde h_3 + \frac{1}{2} [\xi_1,\xi_2]$, \ldots, see \cite{Vizmann}.) 
So far, we defined $g$ only on 
$\bigsqcup_{x \in \bm z} \Sigma \times \Delta^\prime_x \subset X$. 
Recalling (b), $g$ can be extended smoothly by $e \in G$ outside of 
$\bigsqcup_{x \in \bm z} \Sigma \times \Delta^\prime_x \subset X$. 
This extension provides the desired $g \in C^\infty(X,G)$ such that 
$\bm{\iota}^\ast j_X^\ast g = \bm{\iota}^\ast \widetilde h = \widehat h$. 
In particular, by \eqref{bf j vs. j_X} we find $\bm j^\ast g = h$, from which it follows that 
$\bm j^\ast (\null^{g^{-1}}A) = \null^{h^{-1}} (\bm j^\ast A) \in \Omega^1(\Sigma,\k)$,
i.e.\ $\null^{g^{-1}} A$ is an object of $\mathfrak{F}_{\rm bc}(X)$,
and that $(g,e) : \Phi(\null^{g^{-1}} A) \to (A,h)$ is a morphism 
in $\mathfrak{F}(X)$. This completes the proof. 
\end{proof}

In other words, Theorem \ref{thm: equiv Fbc to F} expresses the fact that the gauge field theories
described by the two groupoids $\mathfrak{F}_{\rm bc}(X)$ and $\mathfrak{F}(X)$ are equivalent. 
That is, one may either use fields $A \in \overline{\Omega}^1(X, \g)$ satisfying the strict boundary condition 
$\bm j^\ast A\in \Omega^1(\Sigma, \k)$, or alternatively one may use pairs of fields 
$(A, h) \in \overline{\Omega}^1(X, \g) \times C^\infty(\Sigma, G^{\widehat{\bm z}})$ such that $\bm j^\ast A$ lies in 
$ \Omega^1(\Sigma, \k)$ only up to a gauge transformation determined by the given additional field $h$ on the surface defect $\widehat{D}$. 
The additional field $h \in C^\infty(\Sigma, G^{\widehat{\bm z}})$ living on the surface defect $\widehat{D}$ is called the \emph{edge mode}.

\medskip

Using the equivalence $\Phi$ from Theorem \ref{thm: equiv Fbc to F}, we extend the gauge invariant action 
$S_\omega$ on the groupoid $\mathfrak{F}_{\rm bc}(X)$ to the field groupoid $\mathfrak{F}(X)$
including the edge modes. The extended action 
$S_\omega^{\rm ext}$ on $\mathfrak{F}(X)$ is uniquely determined by 
\begin{subequations} \label{4dCS map em}
\begin{equation} 
S_\omega^{\rm ext} \circ \Phi = S_\omega. 
\end{equation}
Explicitly, for each $(A,h) \in \mathfrak{F}(X)$, we use 
that $\Phi$ is essentially surjective to choose an object 
$\widetilde A \in \mathfrak{F}_{\rm bc}(X)$ and a morphism 
$(g,k) : \Phi(\widetilde A) \to (A,h)$ in $\mathfrak{F}(X)$ and set 
\begin{equation}
S_\omega^{\rm ext}(A,h) \coloneqq S_\omega(\widetilde{A}). 
\end{equation}
\end{subequations}
Using that $\Phi$ is also full, one checks that the above definition 
actually gives a gauge invariant action $S_\omega^{\rm ext}$ on the field groupoid $\mathfrak{F}(X)$. 
In particular, we can choose $k=e$ and $g \in C^\infty(X, G)$ such that $\bm j^\ast g = h$ 
as in the proof of Theorem \ref{thm: equiv Fbc to F} and,
using also \eqref{reg S4d gauge transf} and Propositions \ref{prop: kinetic higher pole} 
and \ref{prop: int WZ compute higher}, we compute $S_\omega^{\rm ext}$ explicitly as
\begin{align} \label{Action edge modes}
S_\omega^{\rm ext}(A, h) = S_\omega(\null^{g^{-1}\!} A)
= S_\omega(A) + \frac{1}{2} \int_\Sigma \big\langle{\mkern-5mu}\big\langle d_\Sigma h h^{-1}, \bm j^\ast A \big\rangle{\mkern-5mu}\big\rangle_{\omega} - \frac{1}{2} \int_{\Sigma \times I} \widehat{h}^\ast \chi_{G^{\widehat{\bm z}}},
\end{align}
where $\widehat{h} \in C^\infty(\Sigma \times I, G^{\widehat{\bm z}})$ 
is any lazy homotopy between $h \in C^\infty(\Sigma, G^{\widehat{\bm z}})$ 
and the constant map $e \in C^\infty(\Sigma, G^{\widehat{\bm z}})$.
The action \eqref{Action edge modes} is to be compared with the action of ordinary $3$-dimensional 
(abelian) Chern-Simons theory given in \cite[(5.1)]{HomEdgeModes}.

\section{Passage to integrable field theories} \label{sec: passage to IFT}
In order to link 4-dimensional Chern-Simons theory to integrable field theories,
we introduce the full subgroupoid $\mathfrak{F}^{1,0,0}(X) \subset \mathfrak{F}(X)$ 
whose objects $(\mathcal L, h) \in \mathfrak{F}^{1,0,0}(X)$ 
are those objects of $\mathfrak{F}(X)$ (cf.\ \eqref{FM def}) which
satisfy the additional condition that $\mathcal L \in \Omega^{1,0,0}(X,\g) \subset \overline{\Omega}^1(X, \g)$ 
is a $(1,0,0)$-form on $X$, i.e.\ $\mathcal{L}$ has no $dz$ and $d\bar z$ components.
Let us mention that morphisms $(g,k): (\mathcal L, h) \to (\null^g \mathcal L, ({\bm j}^\ast g) h k^{-1})$ 
in $\mathfrak{F}^{1,0,0}(X)$ are then given by pairs of maps $(g,k) \in C^\infty(X,G)\times C^\infty(\Sigma, K)$ 
satisfying $\bar \partial g g^{-1} = 0$, which follows from the fact that, by definition, 
also $(\null^g \mathcal L, ({\bm j}^\ast g) h k^{-1})$ lies in $\mathfrak{F}^{1,0,0}(X)$. 
Explicitly, the groupoid introduced above reads as
\begin{equation}\label{F_Sigma groupoid}
\mathfrak{F}^{1,0,0}(X) \coloneqq \left\{
\begin{array}{ll}
\textup{Obj}: & (\mathcal L, h) \in \Omega^{1,0,0}(X,\g)\times C^\infty(\Sigma,G^{\widehat{\bm z}}),
 \;\; \textup{s.t.} \; \null^{h^{-1}} (\bm j^\ast \mathcal{L}) \in \Omega^1(\Sigma,\k),\\
\textup{Mor}: & (g,k): (\mathcal L, h) \to (\null^g \mathcal L, ({\bm j}^\ast g) h k^{-1}), \;\; \\
&\textup{with } \;g \in C^\infty(X, G)\; \text{s.t.}\; \bar\partial g g^{-1}=0 \;\textup{ and }\; k\in C^\infty(\Sigma, K).
\end{array}
\right.
\end{equation}

\begin{remark}
The inclusion functor $ \mathfrak{F}^{1,0,0}(X) \hookrightarrow  \mathfrak{F}(X)$ 
is by definition fully faithful. One might ask if it is also 
essentially surjective, hence an equivalence. 
By direct inspection, it is easy to realise that the answer is positive 
provided that, for each $(A,h) \in \mathfrak{F}(X)$, 
there exists $g \in C^\infty(X,G)$ such that 
$g^{-1} \bar \partial g = A^{0,0,1}$, where 
$A^{0,0,1} \in \Omega^{0,0,1}(X,\g)$ denotes the $(0,0,1)$-component of $A \in \overline{\Omega}^1(X,\g)$. 
In order to simplify the problem, suppose $G = \mathrm{GL}_N(\CC)$,
let us fix a point $a \in \Sigma$ and consider the problem of finding 
such a $g$ on $C_a \coloneqq \{a\} \times C \subset X$. 
Then an argument based on the inverse function theorem for Banach manifolds 
and elliptic regularity, cf.\ \cite[Section 5]{AtiyahBott}, 
shows that the above equation admits local solutions $\{ g_\alpha \}$ 
subordinate to a cover $\{U_\alpha\subseteq C_a\}$ by sufficiently small open subsets of $C_a$. 
As a consequence, $\{ g_{\alpha\beta}\} \coloneqq \{ g_\alpha g_\beta^{-1} \}$ is a \v{C}ech 
1-cocycle on $C_a$ taking values in the sheaf of holomorphic $G$-valued functions. 
The latter is always trivial by \cite[Theorem 30.5]{Foster} 
because $C_a$ is a non-compact Riemann surface. 
This allows us to find a \v{C}ech 0-cochain $\{ h_\alpha \}$ trivialising 
$\{ g_{\alpha\beta} \}$. It follows that setting $g \coloneqq h_\alpha^{-1}g_\alpha$ 
on each $U_\alpha$ defines $g \in C^\infty(C_a,G)$
such that $g^{-1} \bar \partial g = A^{0,0,1}$, as required. 
(Note that, in contrast to $\{g_\alpha\}$, $\{h_\alpha\}$ is holomorphic, 
which is crucial to check that $g$ indeed solves the above equation.) 
Extending this argument to the whole of $X$ 
and for arbitrary $G$ requires to establish smoothly $\Sigma$-parametrised 
analogues with target an arbitrary Lie group $G$ of the arguments 
in \cite[Section 5]{AtiyahBott} and \cite[Theorem 30.5]{Foster}. 
Since essential surjectivity of $\mathfrak{F}^{1,0,0}(X) \hookrightarrow  \mathfrak{F}(X)$ 
is not needed for our constructions below, we shall not further address this issue.
\end{remark}

Since $\mathfrak{F}^{1,0,0}(X) \subset \mathfrak{F}(X)$ is 
a subgroupoid, we can restrict the action on $\mathfrak{F}(X)$ defined in \eqref{4dCS map em} 
to $\mathfrak{F}^{1,0,0}(X)$. From the explicit expression \eqref{Action edge modes}, we obtain
\begin{equation} \label{Action sigma model}
S_\omega^{\rm ext}(\L, h) = \frac{\ii}{4 \pi} \int_X \big( \omega \wedge j^\ast_X \langle \L, \bar\partial \L \rangle \big)_{\rm reg}^{} + \frac{1}{2} \int_\Sigma \big\langle{\mkern-5mu}\big\langle d_\Sigma h h^{-1}, \bm j^\ast \L \big\rangle{\mkern-5mu}\big\rangle_{\omega}^{} - \frac{1}{2} \int_{\Sigma \times I} \widehat{h}^\ast \chi_{G^{\widehat{\bm z}}},
\end{equation}
where the simplification in the first term follows from $\L \in \Omega^{1,0,0}(X, \g)$
by definition of the subgroupoid $\mathfrak{F}^{1,0,0}(X) \subset \mathfrak{F}(X)$, cf.\ \eqref{F_Sigma groupoid}.

\medskip

Let us now derive the Euler-Lagrange equations corresponding to the action \eqref{Action sigma model}.
For this we have to consider variations of objects $(\mathcal{L},h)\in \mathfrak{F}^{1,0,0}(X)$,
i.e.\ variations $(\mathcal{L}',h') = (\mathcal{L} +\epsilon \ell , e^{\epsilon \chi}\, h)$, with
$\ell \in \Omega^{1,0,0}_c(X,\g)$ and $\chi\in C^\infty_c(\Sigma, \g^{\widehat{\bm z}})$,
satisfying the condition $\null^{h'^{-1}} (\bm j^\ast \mathcal{L}') \in \Omega^1(\Sigma,\k)$. 
Expanding this condition to first order in $\epsilon$, one finds that the variations are constrained by
\begin{equation}\label{Variation constraint}
h^{-1} \big( d_\Sigma \chi + [{\bm j}^\ast \mathcal{L},\chi] + {\bm j}^\ast \ell \big) h \in \Omega^1(\Sigma,\k).
\end{equation}
Varying the action \eqref{Action sigma model} and using \eqref{Variation constraint},
one obtains
\begin{align*}
\delta_{(\ell,\chi)}S_\omega^{\rm ext}(\mathcal{L},h) = \frac{\ii}{2\pi} \int_X  \big(\omega \wedge j_X^\ast \langle \ell, \bar\partial \mathcal{L} \rangle \big)_{\rm reg}^{} - \int_\Sigma \big\langle{\mkern-5mu}\big\langle \chi, d_\Sigma( {\bm j}^\ast \mathcal{L}) + \tfrac{1}{2} \big[{\bm j}^\ast \mathcal{L},{\bm j}^\ast \mathcal{L}\big] \big\rangle{\mkern-5mu}\big\rangle_{\omega}^{}.
\end{align*}
From bulk variations, i.e.\ $(\ell,\chi) = (\ell,0)$ with $\mathrm{supp}\, \ell \subset X\setminus D$ 
(note that the constraint \eqref{Variation constraint} is trivially satisfied),
we obtain the equation of motion
\begin{equation*}
\bar\partial \mathcal{L} = 0 \qquad\text{on } X\setminus D.
\end{equation*}
Because $\mathcal{L}\in \Omega^{1,0,0}(X,\g)$ is a smooth $1$-form on $X$,
this equation implies that $\mathcal{L}$ is holomorphic on all of $C$, i.e.\ 
\begin{equation}\label{Lax holomorphic}
\bar\partial \mathcal{L} = 0 \qquad\text{on } X.
\end{equation}
(Recall that $X = \Sigma\times C$, where $C=\CP\setminus {\bm \zeta}$ is the Riemann sphere
with the zeroes of $\omega$ removed. In particular, solutions to \eqref{Lax holomorphic} on $X$
may have poles at ${\bm \zeta} \subset \CP$ with coefficients in $\Omega^{1}(\Sigma, \g)$, 
as required for Lax connections in integrable field theories.)

To study variations with support on the defect, we first observe that, given any
$\chi\in C^\infty_c(\Sigma,\g^{\widehat{\bm z}})$, there exists $\ell \in \Omega^1_c(X,\g)$ such that
the pair $(\ell,\chi)$ satisfies \eqref{Variation constraint}.
Indeed, the equation ${\bm j}^\ast \ell = - d_\Sigma \chi - [{\bm j}^\ast \mathcal{L},\chi ]$
on the jets of $\ell$ can be solved for an arbitrary right hand side by the same method as in the proof of 
Theorem \ref{thm: equiv Fbc to F}. Hence, we obtain the equation of motion
\begin{equation}\label{Bdy equation of motion}
d_\Sigma( {\bm j}^\ast \L) + \tfrac{1}{2}\big[ {\bm j}^\ast \L, {\bm j}^\ast \L\big]=0
\qquad\text{on } \Sigma,
\end{equation}
which means that ${\bm j}^\ast \mathcal{L}\in \Omega^1(\Sigma,\g^{\widehat{\bm z}})$ defines
a flat $G^{\widehat{\bm z}}$-connection on $\Sigma$.

\medskip

To perform the passage to integrable field theories on $\Sigma$, we shall consider
suitable solutions to the bulk equation of motion \eqref{Lax holomorphic} with properties that resemble
those of Lax connections. We will do this in two steps. First, we restrict attention to 
solutions that are meromorphic on $\CP$. Subsequently, 
we will further restrict attention to those solutions for which the defect equation of motion 
\eqref{Bdy equation of motion} can be lifted to a flatness condition for $\L$ on all of $X$. 
(Note that we do \emph{not} solve the defect equation 
of motion \eqref{Bdy equation of motion} on $\Sigma$.)

More precisely, we introduce the following
\begin{definition}\label{OmegaM subspace}
Denoting by $m_y \in \ZZ_{\geq 1}$ the order of the zero $y\in {\bm \zeta}$ 
of $\omega$, we let $\Omega^{r,0,0}_\M(X, \g) \subset \Omega^{r,0,0}(X, \g)$ be
the subspace of those $\g$-valued $(r,0,0)$-forms on $X$ that are 
meromorphic on $\CP$ with poles at each $y \in \bm \zeta$ of order at most $m_y$. 
\end{definition}

Note that, by definition, every $\mathcal{L}\in\Omega^{1,0,0}_\M(X, \g) $
is a solution to the bulk equation of motion \eqref{Lax holomorphic}.
Furthermore, every $\mathcal{L}\in\Omega^{1,0,0}_\M(X, \g) $
can be written explicitly as
\begin{equation}\label{L Ansatz}
\L = \mathcal{L}_{\rm c} + \sum_{y \in \bm \zeta \setminus \{ \infty \}} \sum_{q=0}^{m_y - 1} \frac{\L^y_q}{(z - y)^{q+1}}
+ \sum_{q=0}^{m_\infty-1} \L^\infty_q z^{q+1},
\end{equation}
where $\mathcal{L}_{\rm c}\in \Omega^1(\Sigma,\g)$ and $\mathcal{L}^y_q\in\Omega^1(\Sigma,\g)$, for every
$y\in{\bm \zeta}$ and $q=0,\ldots,m_y-1$, are $\g$-valued 1-forms on $\Sigma$.
Note that the first term of \eqref{L Ansatz} is constant on $\CP$, while the second and third
terms describe the poles at $y\in{\bm \zeta}\setminus\{\infty\}$ and 
at the zero $y=\infty$ of $\omega$, respectively. 

The $1$-form $\L\in \Omega^{1,0,0}_\M(X, \g)$ is still too general to serve as a Lax connection
for integrable field theories. The reason is that the flatness condition,
which is encoded by the defect equation of motion \eqref{Bdy equation of motion}, 
is a priori imposed only for the restriction via ${\bm j}^\ast$ to $\Sigma$ of
(the jets of) the curvature 
$F_\Sigma(\mathcal{L}) \coloneqq d_\Sigma \mathcal{L} +\frac{1}{2}[\mathcal{L},\mathcal{L}] \in\Omega^{2,0,0}(X,\g)$. 
(Note that ${\bm j}^\ast F_\Sigma(\mathcal{L}) = 
d_\Sigma ({\bm j}^\ast\mathcal{L}) +\frac{1}{2}[{\bm j}^\ast \mathcal{L},{\bm j}^\ast \mathcal{L}]$
because ${\bm j}^\ast$ given in \eqref{jet map forms} preserves both the differential $d_\Sigma$
and the Lie bracket $[\cdot,\cdot]$.) In order to upgrade the flatness condition 
from ${\bm j}^\ast F_\Sigma(\mathcal{L}) = 0$ on $\Sigma$ 
(cf.\ \eqref{Bdy equation of motion}) to $F_\Sigma(\mathcal{L})=0$
on $X$, i.e.\ prior to applying ${\bm j}^\ast$, we require the following 
\begin{definition}\label{def: admissible connection}
A form $\mathcal{L}\in \Omega_\M^{1,0,0}(X,\g)$ is called {\em admissible}
if $F_\Sigma(\mathcal{L}) \in\Omega^{2,0,0}_\M(X,\g)$. 
We denote by $\Omega^{1,0,0}_{\rm adm}(X,\g) \subset \Omega_\M^{1,0,0}(X,\g)$ the subspace of admissible forms.
\end{definition}

\begin{example} \label{admissible example}
Note that not every $\mathcal{L}\in \Omega_\M^{1,0,0}(X,\g)$ is admissible,
because the term $[\L,\L]$ in the curvature may have poles at $y\in{\bm \zeta}$
of order greater than $m_y$. A simple algebraic condition which ensures that
$\mathcal{L}$, written in the form \eqref{L Ansatz}, is admissible is given by
\begin{equation*}
\big[\mathcal{L}^y_q, \mathcal{L}^y_{q'}\big] =0,
\end{equation*}
for all $y\in{\bm \zeta}$ and $q,q'$ with $q+q'+2>m_y$. One way to 
achieve this is the following: for each $y \in \bm \zeta$, we introduce a 
coordinate $\sigma_y : \Sigma \to \RR$ on $\Sigma$ and take the $1$-forms 
$\mathcal{L}^y_q \in \Omega^1(\Sigma,\g)$, for $q = 0, \ldots, m_y-1$, 
to be proportional to $d \sigma_y$. 
For example, to produce a Lorentzian integrable field theory, we fix 
a Minkowski metric on $\Sigma$, let $\sigma^\pm$ denote a corresponding pair 
of null coordinates,  choose a subset $\bm \zeta^+ \subset \bm \zeta$ 
and then set $\sigma_y = \sigma^+$ for $y \in \bm \zeta^+$ and 
$\sigma_y = \sigma^-$ for $y \in \bm \zeta \setminus \bm \zeta^+$
in the complement, cf.\ \cite{Delduc:2019bcl}.
\end{example}

\begin{lemma} \label{injectivity jast}
For every $r=0,1,2$, the restriction $\bm j^\ast : \Omega^{r,0,0}_\M(X, \g) \to \Omega^r(\Sigma, \g^{\widehat{\bm z}})$ 
of the morphism \eqref{jet map forms} to the subspace $\Omega^{r,0,0}_\M(X, \g)\subset \Omega^r(X,\g)$ 
introduced in Definition \ref{OmegaM subspace} is injective.
\end{lemma}
\begin{proof}
By definition, any $\eta \in \Omega^{r,0,0}_\M(X, \g)$ is meromorphic 
on $\CP$ with poles at all $y \in \bm \zeta$ of order at most $m_y$ 
and with coefficients in $\Omega^r(\Sigma, \g)$. We need to show that 
if $\iota_x^\ast (\partial_z^p \eta) = 0$, for all $x \in \bm z$
and $p = 0, \ldots, n_x - 1$, then $\eta = 0$.

Consider the polynomial $P(z) \coloneqq \prod_{y \in \bm \zeta \setminus \{ \infty \}} (z - y)^{m_y}$. 
Then $P \eta$ is a polynomial in $z$ of order at most $\sum_{y \in \bm \zeta} m_y$ with coefficients in $\Omega^r(\Sigma, \g)$. 
Since by assumption $\iota_x^\ast (\partial_z^p \eta) = 0$,
for all $x \in \bm z$ and $p = 0, \ldots, n_x - 1$, it follows 
by the Leibniz rule that $\iota_x^\ast (\partial_z^p (P\eta)) = 0$,
for every $x \in \bm z$ and $p = 0, \ldots, n_x - 1$. 
Since $\omega$ is a meromorphic $1$-form on $\CP$,
we have $\sum_{x \in \bm z} n_x = \sum_{y \in \bm \zeta} m_y + 2$, 
which is greater than the degree of the polynomial $P\eta$. 
It follows that $P \eta = 0$ and hence $\eta = 0$.
\end{proof}

\begin{proposition}\label{propo lift flatness}
For any admissible $\mathcal{L} \in \Omega_{\rm adm}^{1,0,0}(X,\g)$, 
the defect equation of motion \eqref{Bdy equation of motion}, 
i.e.\ $\bm j^\ast F_\Sigma(\L) = 0$ on $\Sigma$,  
is equivalent to $F_\Sigma(\L) = 0$ on $X$.
\end{proposition}
\begin{proof}
Suppose $\bm j^\ast F_\Sigma(\L) = 0$. Since $\L$ is admissible, $F_\Sigma(\L) \in \Omega^{2,0,0}_\M(X, \g)$ 
and hence $F_\Sigma(\L) = 0$ by Lemma \ref{injectivity jast}. The converse is obvious.
\end{proof}

The above results motivate us to introduce a suitable subgroupoid of $\mathfrak{F}^{1,0,0}(X)$ 
whose objects $(\L,h)$ are such that $\L\in \Omega_{\rm adm}^{1,0,0}(X,\g)$ is 
admissible in the sense of Definition \ref{def: admissible connection}. In particular, 
such $\L$'s satisfy the bulk equation
of motion \eqref{Lax holomorphic}, are meromorphic on $\CP$ with poles
of the form \eqref{L Ansatz} and, by Proposition \ref{propo lift flatness}, 
the defect equation of motion \eqref{Bdy equation of motion} is equivalent to
flatness $F_{\Sigma}(\L)=0$ on $X$. In other words, such $\L$'s satisfy all
the necessary properties of Lax connections for integrable field theories.
Concerning morphisms $(g,k) : (\L,h)\to (\null^{g}\L,({\bm j}^\ast g)hk^{-1})$
between such objects, by definition of the groupoid $\mathfrak{F}^{1,0,0}(X)$ in \eqref{F_Sigma groupoid}
we have that $g\in C^\infty(X,G)$ is holomorphic on $C$. In order to preserve
the pole structure \eqref{L Ansatz} of admissible $\L$'s under gauge transformations,
we further restrict our attention to those $g$ that are holomorphic on all of $\CP$, and hence constant along $\CP$.
Summing up this discussion, we introduce the following (not necessarily full)
subgroupoid of \eqref{F_Sigma groupoid}
\begin{equation}\label{F_lax groupoid}
\mathfrak{F}_{\rm Lax}(X) \coloneqq \left\{
\begin{array}{ll}
\textup{Obj}: & (\mathcal L, h) \in \Omega_{\rm adm}^{1,0,0}(X,\g)\times C^\infty(\Sigma,G^{\widehat{\bm z}}),
 \;\; \textup{s.t.} \; \null^{h^{-1}} (\bm j^\ast \mathcal{L}) \in \Omega^1(\Sigma,\k),\\
\textup{Mor}: & (g,k): (\mathcal L, h) \to (\null^g \mathcal L, ({\bm j}^\ast g) h k^{-1}), \;\; \\
& \textup{with } \;g \in C^\infty(\Sigma,G) \;\textup{ and }\; k\in C^\infty(\Sigma, K),
\end{array}
\right.
\end{equation}
where we are implicitly identifying a map $g \in C^\infty(\Sigma,G)$ with its pullback along 
the projection $p_\Sigma : X \to \Sigma$. Under this identification, we have that ${\bm j}^\ast g = \Delta(g)$, where  
$\Delta : G\to G^{\widehat{\bm z}}\,,~g\mapsto (g)_{x\in{\bm z}}$ is the diagonal map
to the defect group \eqref{defect group higher}.

With these preparations, we are now ready to describe how $2$-dimensional
integrable field theories arise from $4$-dimensional Chern-Simons theory.
Consider the groupoid
\begin{equation}\label{F_2d groupoid}
\mathfrak{F}_{\rm 2d}(\Sigma) \coloneqq \left\{
\begin{array}{ll}
\textup{Obj}: & h \in C^\infty(\Sigma,G^{\widehat{\bm z}}),\\
\textup{Mor}: & (g,k): h \to  \Delta(g) h k^{-1},  \;\; \\
& \textup{with } \;g \in C^\infty(\Sigma,G) \;\textup{ and }\; k\in C^\infty(\Sigma, K),
\end{array}
\right.
\end{equation}
of $G^{\widehat{\bm z}}$-valued fields on $\Sigma$ and note that there exists a forgetful functor
\begin{equation}\label{projection functor}
\pi : \mathfrak{F}_{\rm Lax}(X) \longrightarrow \mathfrak{F}_{\rm 2d}(\Sigma)
\end{equation}
that sends an object $(\L,h)$ to $h$ and a
morphism $(g,k): (\mathcal L, h) \to (\null^g \mathcal L, ({\bm j}^\ast g) h k^{-1})$
to $(g,k) :  h \to  \Delta(g) h k^{-1}$. 
If this functor was fully faithful and essentially surjective, i.e.\ 
an equivalence, then we could transfer the action
\eqref{Action sigma model} to an action
\begin{equation}\label{eqn: 2d action 1}
S_\omega^{\rm 2d} \coloneqq S_\omega^{\rm ext} \circ \pi^{-1}
\end{equation}
defined on the groupoid $\mathfrak{F}_{\rm 2d}(\Sigma)$ in \eqref{F_2d groupoid}, 
where $\pi^{-1} : \mathfrak{F}_{\rm 2d}(\Sigma) \to \mathfrak{F}_{\rm Lax}(X)$
denotes a quasi-inverse of $\pi$. 
Gauge invariance of $S_\omega^{\rm ext}$ entails that $S_\omega^{\rm 2d}$ does not depend
on the choice of quasi-inverse.
While \eqref{projection functor} is clearly a faithful functor, fullness and essential surjectivity
do not appear to be automatic. These properties of the functor $\pi$ can be related to
existence and uniqueness of solutions $\L\in\Omega^{1,0,0}_{\rm adm}(X,\g)$ 
for a fixed  $h\in C^\infty(\Sigma,G^{\widehat{\bm z}})$ to the condition $\null^{h^{-1}}({\bm j}^\ast \L) \in \Omega^1(\Sigma,\k)$
on objects $(\L,h)$ of $\mathfrak{F}_{\rm Lax}(X)$, cf.\ \eqref{F_lax groupoid}.
\begin{proposition}
The functor $\pi$ in \eqref{projection functor} is essentially surjective 
if and only if it is surjective on objects, i.e.\ for each $h\in C^\infty(\Sigma,G^{\widehat{\bm z}})$ 
there exists $\L\in\Omega^{1,0,0}_{\rm adm}(X,\g)$ such that $(\L,h)\in \mathfrak{F}_{\rm Lax}(X)$.
It is full if and only if for each $h\in C^\infty(\Sigma,G^{\widehat{\bm z}})$ 
there exists at most one object of the form $(\L,h)$ in $\mathfrak{F}_{\rm Lax}(X)$.
\end{proposition}
\begin{proof}
For the first statement, the implication ``$\Leftarrow$'' is obvious. To prove
the implication ``$\Rightarrow$'', let us assume that $\pi$ is essentially surjective.
Then there exists, for each $h\in C^\infty(\Sigma,G^{\widehat{\bm z}})$,
an object $(\L',h')$ in  $ \mathfrak{F}_{\rm Lax}(X)$ and a morphism
$(g,k) : h\to h' = \pi(\L',h')$ in $\mathfrak{F}_{\rm 2d}(\Sigma)$. 
Setting $\L \coloneqq \null^{g^{-1}}\L' \in\Omega^{1,0,0}_{\rm adm}(X,\g)$,
we obtain
\begin{equation*}
\null^{h^{-1}}({\bm j}^\ast \L) = \null^{h^{-1}\Delta(g^{-1})}({\bm j}^\ast \L') = 
\null^{k^{-1}h'^{-1}}({\bm j}^\ast \L') \in \Omega^1(\Sigma,\k),
\end{equation*}
where in the second step we used $h' = \Delta(g) h k^{-1}$. The
last step then follows from $\null^{h'^{-1}}({\bm j}^\ast \L') \in \Omega^1(\Sigma,\k)$,
as $(\L',h')$ is by hypothesis an object in  $ \mathfrak{F}_{\rm Lax}(X)$,
and the fact that $k\in C^\infty(\Sigma,K)$ is a map to the subgroup $K\subset G^{\widehat{\bm z}}$.

Let us consider now the second statement.
We prove the implication ``$\Rightarrow$'' by contraposition.
Suppose that there exist objects $(\L,h),(\L',h)$ in $ \mathfrak{F}_{\rm Lax}(X)$
such that $\L'\neq \L$. Then there does not exist a morphism
$(\L,h)\to (\L',h)$ in $ \mathfrak{F}_{\rm Lax}(X)$ that
maps under $\pi$ to the identity $\id : h\to h$ in $ \mathfrak{F}_{\rm 2d}(\Sigma)$, 
hence $\pi$ is not full. To prove the implication ``$\Leftarrow$'',
let $(\L,h),(\L',h')$ be arbitrary objects in $ \mathfrak{F}_{\rm Lax}(X)$
and consider any morphism $(g,k) : h\to h'$ in $ \mathfrak{F}_{\rm 2d}(\Sigma)$.
We define the morphism $(g,k) : (\L,h)\to (\null^g \L,h')$ in
$ \mathfrak{F}_{\rm Lax}(X)$ and observe that by hypothesis $\null^{g}\L =\L'$.
Hence, we obtain a morphism $(g,k) : (\L,h)\to (\L',h')$ in
$ \mathfrak{F}_{\rm Lax}(X)$ and thereby prove that $\pi$ is full.
\end{proof}

\begin{corollary} \label{cor: projection functor equivalence}
The functor $\pi$ in \eqref{projection functor} is an equivalence of groupoids if
and only if for each $h\in C^\infty(\Sigma,G^{\widehat{\bm z}})$ 
there exists a unique $\L\in\Omega^{1,0,0}_{\rm adm}(X,\g)$ such that $(\L,h)\in \mathfrak{F}_{\rm Lax}(X)$,
i.e.\ such that $\null^{h^{-1}}({\bm j}^\ast \L) \in \Omega^1(\Sigma,\k)$.
\end{corollary}

\begin{remark}
Let us note that whether or not the functor $\pi$ in \eqref{projection functor} is an equivalence will depend on the choice of
isotropic subalgebra $\k \subset \g^{\widehat{\bm z}}$ used to impose boundary conditions at the surface 
defects in \S\ref{sec: boundary conditions}. Examples of suitable choices when $n_x \leq 2$ for all $x \in \bm z$ 
can be found in \cite{Costello:2019tri, Delduc:2019whp}.
In light of the present work, the problem of classifying isotropic subalgebras $\k \subset \g^{\widehat{\bm z}}$ 
for which the condition $\null^{h^{-1}}({\bm j}^\ast \L) \in \Omega^1(\Sigma,\k)$ admits a unique solution for 
$\L$ in terms of $h$ is an important one in view of the broader open problem of classifying $2$-dimensional 
integrable field theories.
\end{remark}

Suppose now that the functor $\pi$ in \eqref{projection functor} is an equivalence.
Using Corollary \ref{cor: projection functor equivalence}, we can then construct a \emph{strict} inverse
\begin{equation*}
\pi^{-1} :   \mathfrak{F}_{\rm 2d}(\Sigma) \longrightarrow \mathfrak{F}_{\rm Lax}(X).
\end{equation*}
This functor sends an object $h$ to $(\L(h),h)$, where $\L(h)\in\Omega^{1,0,0}_{\rm adm}(X,\g)$
is the unique element such that $(\L(h),h)$ is an object in $\mathfrak{F}_{\rm Lax}(X)$.
To a morphism $(g,k) : h\to h'= \Delta(g)hk^{-1}$ in $\mathfrak{F}_{\rm 2d}(\Sigma)$, this functor assigns the 
morphism $(g,k) : (\L(h),h)\to (\L(h'),h')$ in $\mathfrak{F}_{\rm Lax}(X)$,
where $\L(h') = \L(\Delta(g)hk^{-1}) = \null^{g} \L(h)$ by the uniqueness of 
Corollary \ref{cor: projection functor equivalence}. 
Using this description of $\pi^{-1}$, we obtain an explicit expression for the
action in \eqref{eqn: 2d action 1} 
\begin{equation}\label{eqn: 2d action}
S_{\omega}^{\rm 2d} (h) = S_\omega^{\rm ext}(\L(h), h) =
\frac{1}{2} \int_\Sigma \big\langle{\mkern-5mu}\big\langle d_\Sigma h h^{-1}, \bm j^\ast \L(h) \big\rangle{\mkern-5mu}\big\rangle_{\omega}^{} - \frac{1}{2} \int_{\Sigma \times I} \widehat{h}^\ast \chi_{G^{\widehat{\bm z}}},
\end{equation}
where the first term in \eqref{Action sigma model} vanishes because $\bar\partial \L(h) =0$
by definition of the groupoid $\mathfrak{F}_{\rm Lax}(X) $ in \eqref{F_lax groupoid}.
We would like to emphasise that the action \eqref{eqn: 2d action} is for a 
$G^{\widehat{\bm z}}$-valued field $h$ living on the $2$-dimensional manifold $\Sigma$ 
and that it describes an integrable field theory with Lax connection $\L(h)$. 
Furthermore, the action $S_{\omega}^{\rm 2d}$ is by construction
gauge invariant under the morphisms of the groupoid $ \mathfrak{F}_{\rm 2d}(\Sigma)$ 
introduced in \eqref{F_2d groupoid}.

\begin{remark}\label{minimal}
There is a more minimalistic procedure for transferring the action
$S_\omega^{\rm ext}$ (cf.\ \eqref{4dCS map em}) on the subgroupoid $\mathfrak{F}_{\rm Lax}(X) \subset \mathfrak{F}(X)$ 
to an action $S_{\omega}^{\rm 2d}$ on $\mathfrak{F}_{\rm 2d}(\Sigma)$ 
along the functor $\pi : \mathfrak{F}_{\rm Lax}(X) \to \mathfrak{F}_{\rm 2d}(\Sigma)$ 
in \eqref{projection functor}, which only requires the latter 
to be essentially surjective and not necessarily full. 
This is based on the following observation. 
The datum of a gauge invariant action $S_\omega^{\rm 2d}$ 
on the groupoid $\mathfrak{F}_{\rm 2d}(\Sigma)$ 
is equivalent to the datum of a function $S_\omega^{\rm 2d}$ 
on the set $\pi_0(\mathfrak{F}_{\rm 2d}(\Sigma))$ 
of isomorphism classes of objects. 
Furthermore, essential surjectivity of the functor 
$\pi : \mathfrak{F}_{\rm Lax}(X) \to \mathfrak{F}_{\rm 2d}(\Sigma)$ 
is equivalent to surjectivity of the induced map 
$\pi: \pi_0(\mathfrak{F}_{\rm Lax}(X)) \to \pi_0(\mathfrak{F}_{\rm 2d}(\Sigma))$ 
between sets of isomorphism classes. 
Therefore, in order to transfer $S_\omega^{\rm ext}$ 
to $\mathfrak{F}_{\rm 2d}(\Sigma)$, 
we can choose a section $\sigma$ of the surjective map 
$\pi: \pi_0(\mathfrak{F}_{\rm Lax}(X)) \to \pi_0(\mathfrak{F}_{\rm 2d}(\Sigma))$ 
and define $S_\omega^{\rm 2d} \coloneqq S_\omega^{\rm ext} \circ \sigma$. 
More generally, we can choose a suitable measure $w$ 
on the set of sections $\sigma$ and define $S_\omega^{\rm 2d}$ as the 
$w$-average over all sections $\sigma$ of $S_\omega^{\rm ext} \circ \sigma$. 
(For a fixed section $\sigma$, the Dirac measure $w = \delta_\sigma$ 
recovers the construction considered previously in this remark.)
We stress, however, that this alternative construction of $S_\omega^{\rm 2d}$ 
in general depends on the choice of measure $w$ on the set of sections $\sigma$. 
Whenever $\pi$ is both essentially surjective and full, 
$\pi: \pi_0(\mathfrak{F}_{\rm Lax}(X)) \to \pi_0(\mathfrak{F}_{\rm 2d}(\Sigma))$ 
is actually bijective and hence $S_\omega^{\rm 2d} \coloneqq S_\omega^{\rm ext} \circ \pi^{-1}$ 
is uniquely determined (there is exactly one section $\sigma = \pi^{-1}$). 
In particular, the construction of 
$S_\omega^{\rm 2d}$ presented before this remark agrees with the one considered here.

When $\pi$ is essentially surjective but not full, however, it becomes more 
difficult to interpret the output of our construction as an integrable field 
theory since the candidate Lax connection $\L$ in general fails to be 
\emph{uniquely} determined by the field $h$ living on $\Sigma$. 
\end{remark}

\end{document}